\documentclass[aps,pra,reprint,superscriptaddress,nofootinbib,longbibliography,floatfix]{revtex4-2}
\usepackage{amsmath,amssymb,amsthm,bm,graphicx,microtype,needspace,xcolor}
\usepackage[colorlinks=true,linkcolor=blue,citecolor=blue,urlcolor=blue]{hyperref}
\makeatletter
\AtBeginDocument{\immediate\write\@auxout{\string\citation{apsrev42Control}}}
\makeatother
\newcommand{\ket}[1]{\left|#1\right\rangle}
\newcommand{\bra}[1]{\left\langle#1\right|}
\newcommand{\Tr}{\operatorname{Tr}}
\newtheorem{theorem}{Theorem}

\newtheorem{corollary}[theorem]{Corollary}

\begin{document}
\title{Ergotropy of nonequilibrium steady states beyond the full secular approximation}

\author{Xing-Xing Hu}
\affiliation{School of Physics, Northwest University, Xi'an 710127, China}

\author{Hui-Yu Yang}
\affiliation{School of Physics, Northwest University, Xi'an 710127, China}

\author{Kun Zhang}
\email{kunzhang@nwu.edu.cn}
\affiliation{School of Physics, Northwest University, Xi'an 710127, China}
\affiliation{Shaanxi Key Laboratory for Theoretical Physics Frontiers, Xi'an 710127, China}
\affiliation{Peng Huanwu Center for Fundamental Theory, Xi'an 710127, China}
\affiliation{Fundamental Discipline Research Center for Quantum Science and Technology of Shaanxi Province, Xi'an 710127, China}

\author{Xiao-Hui Wang}
\email{xhwang@nwu.edu.cn}
\affiliation{School of Physics, Northwest University, Xi'an 710127, China}
\affiliation{Shaanxi Key Laboratory for Theoretical Physics Frontiers, Xi'an 710127, China}
\affiliation{Peng Huanwu Center for Fundamental Theory, Xi'an 710127, China}
\affiliation{Fundamental Discipline Research Center for Quantum Science and Technology of Shaanxi Province, Xi'an 710127, China}

\date{\today}
\begin{abstract}
Understanding the generation and activation of ergotropy, the maximum work extractable by cyclic unitary operations, is central to quantum battery research. Positive-temperature thermal states are passive, whereas nonequilibrium environments can sustain nonpassive steady states with nonzero ergotropy. However, the full secular approximation commonly used in Lindblad descriptions can suppress steady-state coherence and obscure its contribution to ergotropy. Using a Bloch--Redfield master equation beyond the full secular approximation, we investigate the ergotropy of the nonequilibrium steady state (NESS) of two exchange-coupled qubits connected to independent thermal reservoirs. The NESS ergotropy of two qubits with the same frequency is entirely coherent and grows quadratically with temperature bias near equilibrium. In contrast, detuned qubits can exhibit incoherent ergotropy. We derive an exact analytical criterion for its activation and show that activation requires the hotter reservoir to be coupled to the qubit with the higher bare frequency. Comparing NESSs obtained with and without full secularization reveals quantitative and qualitative differences in ergotropy that reflect changes in both its coherent and incoherent contributions. The results clarify the mechanisms of NESS ergotropy activation and demonstrate the importance of retaining steady-state coherence beyond the full secular approximation when characterizing stationary work resources.
\end{abstract}
\maketitle

\section{Introduction}
\label{sec:introduction}
Quantum batteries provide a setting for exploiting quantum coherence, correlations, and collective control in energy storage and work extraction~\cite{Campaioli2024}. Their performance depends on how much of the stored energy can be delivered as useful work. For a fixed Hamiltonian, a passive state yields no work under cyclic unitary operations~\cite{PuszWoronowicz1978,Lenard1978}. Ergotropy quantifies the maximum work extractable through such operations and therefore distinguishes useful stored energy from energy trapped in a passive state~\cite{Allahverdyan2004}. It provides an operational benchmark for comparing battery states and identifying the resources responsible for work extraction~\cite{Skrzypczyk2014,Korzekwa2016}.

The coherent--incoherent decomposition of ergotropy separates the work available from energy populations from the additional contribution lost upon dephasing in the energy eigenbasis~\cite{Francica2020,ChinPhysLett.43.060602}. Studies have clarified the roles of coherence, correlations, and control constraints in work extraction~\cite{Shi2022,Lobejko2022,Alicki2013,Hovhannisyan2013,PerarnauLlobet2015,Andolina2019,Aberg2014,Lostaglio2015NC,Lostaglio2015PRX,PhysRevB.109.235432}. The coherent and incoherent contributions have also been explored experimentally and under environmental noise~\cite{Niu2024,Gao2026,HadipourHaseli2024}. Related work links coherence to irreversible work, entropy production, and heat-engine power~\cite{Francica2019,Santos2019,Scully2011,Brandner2017}. 

Maintaining ergotropy in an open quantum system requires sustaining a nonpassive state against relaxation. Thermalization to a positive-temperature Gibbs state eliminates ergotropy. A natural route to generating or sustaining ergotropy is through external driving or engineered system--environment interactions~\cite{Farina2019,Barra2019,Guarnieri2020,Cakmak2020,PhysRevA.103.033715,Li2022,MayoRoncaglia2022,Huang2022,DowningUkhtary2023,Shastri2025,WangYuWang2025,Xu2026TightBinding}. Another approach is to impose a temperature bias by coupling different subsystems or transitions to reservoirs at different temperatures. The resulting nonequilibrium steady state (NESS) can sustain positive ergotropy~\cite{Xu2019,Tacchino2020,Santos2021,Leitch2024,Mojaveri2024,SarmahGoswami2024,Boeyens2025,HadipourHaseli2024SteadyWork,HadipourHaseli2025BathMod,Yang2026}. 

Although Lindblad master equations are widely used to describe open quantum systems~\cite{Lindblad1976,Gorini1976,Breuer2002}, the full secular approximation commonly employed in their weak-coupling derivation can yield inaccurate NESS predictions when distinct transition frequencies are not well separated~\cite{Davies1974,Eastham2016,Cattaneo2019,Trushechkin2021,Potts2021,Tupkary2022}. In particular, full secularization can eliminate energy-basis coherences sustained by nonequilibrium environments~\cite{LiYang2014,LiCaiSun2015,Huangfu2018,Dodin2018,Koyu2021,Ivander2022,Donati2026}. Comparisons with fully secular predictions have revealed quantitative and even qualitative differences in stationary quantum correlations~\cite{Gonzalez2017,Wang2019,WangWang2019,DAbbruzzo2024}. Other NESS resources and their applications beyond full secularization have also been investigated, including steering and Bell nonlocality~\cite{ZhangSteering2021,ZhangBell2021}, temporal quantum correlations~\cite{Zhang2020}, quantum metrology~\cite{Wang2018Metrology,WangWang2022,Abouelkhir2026}, and quantum teleportation~\cite{Yan2025Teleportation,Yan2026QET}. Since coherence is closely linked to ergotropy~\cite{Francica2020,ChinPhysLett.43.060602}, these findings motivate a systematic examination of how full secularization affects NESS ergotropy, its coherent and incoherent contributions, and their activation mechanisms.

In this work, we study the NESS ergotropy obtained without the full secular approximation in a model of two qubits coupled to independent thermal reservoirs. For nonzero exchange coupling between two qubits, the NESS has nonzero coherence in the energy eigenbasis whenever the temperature bias is nonzero. We find that resonant qubits with the same bare frequency have only coherent ergotropy, which grows quadratically with temperature bias near equilibrium and linearly with exchange coupling in the weak-exchange limit at fixed nonzero temperature bias. In contrast, finite detuning allows incoherent ergotropy, for which we derive an exact activation criterion and show that the hotter reservoir must couple to the higher-frequency bare qubit. Finally, we compare NESS ergotropy with and without the full secular approximation and clarify how this approximation affects both the coherent and incoherent contributions.

The rest of the manuscript is organized as follows. Section~\ref{sec:model} introduces the model, dynamics, and analytical NESS.
Section~\ref{sec:steady_ergotropy} presents an analysis of NESS ergotropy and its decomposition.
Sections~\ref{sec:resonance} and~\ref{sec:detuning} analyze the resonant and detuned
cases, including ergotropy activation and optimization.
Section~\ref{sec:secular_comparison} compares predictions with and without the
full secular approximation, and Sec.~\ref{sec:conclusion} concludes. Analytical derivations are provided in
Appendixes~\ref{app:br_steady_state}--\ref{app:finite_detuning_analysis}.

\section{Model, nonsecular dynamics, and steady state}
\label{sec:model}

In this section, we set up the model and present the analytical NESS without the full secular approximation. Subsection~\ref{subsec:spectrum_transitions} introduces the model, its energy
eigenstates, and the transition frequencies.
Subsection~\ref{subsec:br_equation} specifies the local thermal reservoirs and
the partial-secular Bloch--Redfield equation.
Subsection~\ref{subsec:compact_param} presents the analytical NESS and the
parametrization used in the subsequent ergotropy analysis.

\subsection{Hamiltonian and spectrum}
\label{subsec:spectrum_transitions}


\begin{figure}[t]
\centering
\includegraphics[width=0.98\columnwidth]{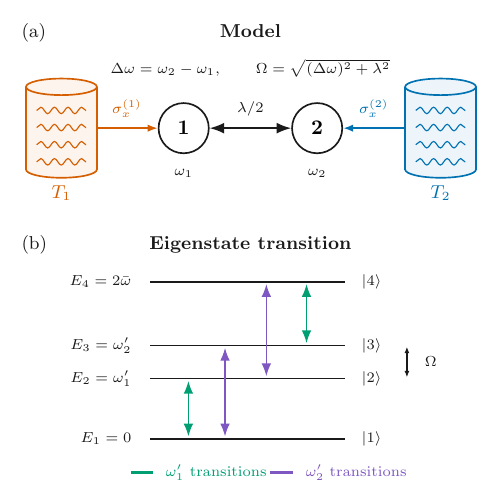}
\caption{Schematic illustration of two interacting qubits coupled to two
local thermal reservoirs.
(a) Bosonic reservoirs at $T_1,T_2$ couple locally to qubits
with bare frequencies $\omega_1,\omega_2$. The exchange amplitude is
$\lambda/2$. (b) The four eigenstates and the corresponding transitions of
the two-qubit model. The eigenenergies are given by
Eq.~\eqref{eq:transition_frequencies}. The two arrow colors group bath-induced
transitions with frequencies $\omega'_1$ and $\omega'_2$.}
\label{fig:model_mechanism}
\end{figure}

We consider two qubits with bare frequencies $\omega_1,\omega_2$ coupled by
an exchange interaction~\cite{PhysRevA.106.032212}, with Hamiltonian
\begin{equation}
H_S={}\omega_1\sigma_+^{(1)}\sigma_-^{(1)}
+\omega_2\sigma_+^{(2)}\sigma_-^{(2)}+\frac{\lambda}{2}\left(\sigma_+^{(1)}\sigma_-^{(2)}
+\sigma_-^{(1)}\sigma_+^{(2)}\right),
\label{eq:HS}
\end{equation}
with exchange strength $\lambda\ge0$. Here $\ket{g_j}$ and $\ket{e_j}$ are the
ground and excited states of qubit $j=1,2$, respectively, with
$\sigma_+^{(j)}=\ket{e_j}\bra{g_j}$ and
$\sigma_-^{(j)}=(\sigma_+^{(j)})^\dagger$. We set $\hbar=k_B=1$. This
Hamiltonian describes, for example, the effective interaction used for
coherent state transfer between two superconducting qubits dispersively
coupled to a microwave cavity~\cite{Majer2007}. If it is obtained from a
transverse qubit--qubit interaction within the rotating-wave approximation,
the omitted counter-rotating terms create or
annihilate two excitations and oscillate at the sum frequency
$\omega_1+\omega_2=2\bar\omega$. They remain perturbative only when
$\lambda\ll\omega_1+\omega_2$. At ultrastrong coupling the
excitation-number-conserving Hamiltonian in Eq.~\eqref{eq:HS} is no longer
adequate~\cite{Blais2021}. Figure~\ref{fig:model_mechanism} summarizes the
model.

The interacting Hamiltonian has the energy spectrum
\begin{equation}
(E_1,E_2,E_3,E_4)=(0,\omega'_1,\omega'_2,2\bar\omega),
\label{eq:transition_frequencies}
\end{equation}
where
\begin{equation}
\omega'_1=\bar\omega-\frac{\Omega}{2},\qquad
\omega'_2=\bar\omega+\frac{\Omega}{2},\qquad \bar\omega=\frac{\omega_1+\omega_2}{2}.
\end{equation}
Here $\Omega$ is the energy splitting between the two middle levels, namely
\begin{equation}
\Omega=\sqrt{(\Delta\omega)^2+\lambda^2},\qquad 
\Delta\omega=\omega_2-\omega_1. 
\label{eq:Delta_Omega}
\end{equation}
The effective exchange model is intended for
$\lambda\ll2\bar\omega$~\cite{Blais2021}. Therefore we only consider
$0<\Omega<2\bar\omega$, which keeps both transition frequencies
positive and fixes the level order $E_1<E_2<E_3<E_4$. The resulting spectrum
is shown in Fig.~\ref{fig:model_mechanism}.

The corresponding eigenvectors are given by 
\begin{subequations}
\label{eq:eigenstates}
\begin{align}
&\ket{1}=\ket{gg},\\
&\ket{2}=\cos\frac{\theta}{2}\ket{eg}
-\sin\frac{\theta}{2}\ket{ge},\\
&\ket{3}=\sin\frac{\theta}{2}\ket{eg}
+\cos\frac{\theta}{2}\ket{ge},\\
&\ket{4}=\ket{ee}.
\end{align}
\end{subequations}
The mixing angle is defined by 
\begin{equation}
\cos\theta=\frac{\Delta\omega}{\Omega},\qquad
\sin\theta=\frac{\lambda}{\Omega},\qquad 0\leq\theta\leq\pi.
\label{eq:theta_definition}
\end{equation}
At resonance, $\Delta\omega=0$ and $\lambda>0$ imply $\theta=\pi/2$, so the
two single-excitation eigenstates become equal-weight antisymmetric and
symmetric superpositions of $\ket{eg}$ and $\ket{ge}$, corresponding to
maximally entangled two-qubit states.

\subsection{Bloch--Redfield equation}
\label{subsec:br_equation}

We model the two bosonic reservoirs by the free Hamiltonians
\begin{equation}
H_{B_j}=\sum_k\nu_{jk}b_{jk}^\dagger b_{jk},
\label{eq:HB}
\end{equation}
where $\nu_{jk}>0$ is the frequency of mode $k$ in reservoir $j=1,2$, and
$b_{jk}$ ($b_{jk}^\dagger$) is the corresponding bosonic annihilation
(creation) operator. We use the standard transverse spin--boson interaction,
in which each qubit couples to an independent thermal
reservoir~\cite{Breuer2002,Schaller2014}:
\begin{align}
H_{SB_j}=\sigma_x^{(j)}\otimes
\sum_k(g_{jk}b_{jk}+g_{jk}^*b_{jk}^\dagger),
\label{eq:HSB}
\end{align}
where $g_{jk}$ is the complex coupling amplitude and
$\sigma_x^{(j)}=\sigma_+^{(j)}+\sigma_-^{(j)}$. The total Hamiltonian is
\begin{equation}
H=H_S+\sum_j(H_{B_j}+H_{SB_j}),    
\end{equation}
where $H_S$, $H_{B_j}$, and $H_{SB_j}$ are specified by
Eqs.~\eqref{eq:HS}, \eqref{eq:HB}, and \eqref{eq:HSB}.

We take the two reservoirs to be mutually independent, stationary, and
initially thermal at temperatures $T_j$, with an initially factorized
system--reservoir state. To derive the reduced equation from the total
Hamiltonian $H$, we expand the interaction-picture dynamics to second order
in $H_{SB_j}$ (Born approximation), replace the delayed reduced state by its
instantaneous value, and extend the bath-memory integral to infinity (Markov
approximation)~\cite{Breuer2002,Schaller2014}. Decomposing each system coupling
operator into frequency components separates rapidly and slowly rotating
contributions: terms mixing positive and negative frequencies are discarded,
whereas interference between distinct transitions of the same sign is
retained. This partial-secular step couples population and coherence dynamics
and therefore permits nonzero steady-state coherence~\cite{Eastham2016}.
Returning to the Schr\"odinger picture gives
\begin{equation}
\dot\rho=-i[H_S,\rho]+\sum_j\mathcal R_j(\rho),
\label{eq:BR_master_equation}
\end{equation}
where $i^2=-1$ and $\rho$ is the reduced two-qubit state. The contribution
from reservoir $j$ is
\begin{multline}
\mathcal R_j(\rho)
={}\sum_{\substack{
\omega,\omega'\in\{\pm\omega'_1,\pm\omega'_2\}\\
\omega\omega'>0}}
\Bigl(\Gamma_j(\omega)
 [A_j(\omega)\rho,A_j^\dagger(\omega')]\\
+\Gamma_j^*(\omega)
 [A_j(\omega'),\rho A_j^\dagger(\omega)]\Bigr),
\label{eq:Redfield_dissipator}
\end{multline}
where $A_j(\omega)$ is the transition operator and $\Gamma_j(\omega)$ is the
corresponding rate. The restriction $\omega\omega'>0$ retains only frequency
pairs with the same sign.

The energy-basis transition components of $\sigma_x^{(j)}$ are
\begin{equation}
A_j(\omega)=\sum_{E_m-E_n=\omega}
\ket n\!\bra n\,\sigma_x^{(j)}\,\ket m\!\bra m,
\label{eq:jump_operator_definition}
\end{equation}
with $A_j(-\omega)=A_j^\dagger(\omega)$. The nonzero positive-frequency
operators are
\begin{subequations}
\label{eq:jump_operators}
\begin{align}
A_1(\omega'_1)&=\cos\frac{\theta}{2}
(\ket1\!\bra2+\ket3\!\bra4),\\
A_1(\omega'_2)&=\sin\frac{\theta}{2}
(\ket1\!\bra3-\ket2\!\bra4),\\
A_2(\omega'_1)&=\sin\frac{\theta}{2}
(-\ket1\!\bra2+\ket3\!\bra4),\\
A_2(\omega'_2)&=\cos\frac{\theta}{2}
(\ket1\!\bra3+\ket2\!\bra4).
\end{align}
\end{subequations}
The relative signs arise because the two reservoirs couple to different local qubits. In the energy eigenbasis, they assign different phases to transition
amplitudes through the eigenstates $\ket2$ and $\ket3$, and therefore
enter the cross-frequency interference that sustains steady-state coherence.

The rate coefficients $\Gamma_j(\omega)$ in
Eq.~\eqref{eq:Redfield_dissipator} follow from the thermal correlation
functions. After neglecting their principal-value parts, the coefficients are
real and take the form
\begin{equation}
\Gamma_j(\omega'_a)=J(\omega'_a)(1+n_j(\omega'_a)),\quad
\Gamma_j(-\omega'_a)=J(\omega'_a)n_j(\omega'_a).
\label{eq:Gamma_rate_convention}
\end{equation}
Here
\begin{equation}
n_j(\omega)=\frac{1}{e^{\omega/T_j}-1}
\label{eq:bose_occupation}
\end{equation}
is the Bose occupation of reservoir $j$ at temperature $T_j$. We assume
identical coupling spectra for the two reservoirs. The frequency-dependent bath coupling strength at a positive Bohr frequency is
given by
\begin{equation}
J(\omega'_a)=\pi\sum_k|g_{jk}|^2\delta(\omega'_a-\nu_{jk}),
\label{eq:flat_spectrum}
\end{equation}
with $j,a=1,2$, where $j$ labels the reservoir, $a$ labels the transition
frequency, and $\delta(\cdot)$ is the Dirac delta distribution. Since
$\omega'_{1,2}=\bar\omega\mp\Omega/2$, the two transitions probe a frequency
window of width $\Omega$. When the bath coupling profile varies on a much
larger frequency scale,
one can use
$J(\omega'_1)\simeq J(\omega'_2)=J$. This flat-spectrum approximation is
standard in weak-coupling Markovian treatments of smooth broadband
reservoirs~\cite{Breuer2002,Schaller2014}. Moreover, the steady-state
predictions were found to be insensitive to the detailed spectral shape in
the Markovian regime~\cite{Wang2019}.

The total dissipator can be separated as
\begin{equation}
\sum_j\mathcal R_j(\rho)
=\mathcal D_0(\rho)+\mathcal D_s(\rho),
\label{eq:D0_Ds_split}
\end{equation}
where $\mathcal D_0$ collects the contributions with
$\omega=\omega'$, whereas $\mathcal D_s$ collects the cross-frequency
contributions with $\omega\ne\omega'$ and $\omega\omega'>0$. These
cross-frequency terms describe indistinguishable relaxation or excitation
paths and can sustain steady-state coherence and alter population
relaxation~\cite{Eastham2016}. Full secularization gives $\mathcal D_s=0$.
Thus Eq.~\eqref{eq:BR_master_equation} is a partial-secular Redfield equation
in the interacting energy basis, rather than a local Lindblad equation.

Although the partial-secular Redfield equation retains cross-frequency effects
that can be important out of equilibrium, its dynamics does not guarantee
complete positivity. Different approaches can restore positivity
to the Redfield equation~\cite{FarinaGiovannetti2019,NathanRudner2020,
Trushechkin2021,Potts2021}. Accordingly, we restrict the reported results to
parameter values for which the stationary solution is positive semidefinite.

\subsection{Analytical nonequilibrium steady state}
\label{subsec:compact_param}

From the Bloch--Redfield equation in
Eq.~\eqref{eq:BR_master_equation}, we obtain the NESS analytically. In the
energy eigenbasis it takes the form
\begin{equation}
\rho^{\rm ss}=\begin{pmatrix}
p_1&0&0&0\\0&p_2&c&0\\0&c^*&p_3&0\\0&0&0&p_4
\end{pmatrix}.
\label{eq:ss_block}
\end{equation}
Here the diagonal entries are the energy-eigenstate populations, and the only
nonvanishing coherence is $c=\bra{2}\rho^{\rm ss}\ket{3}$.

The generator preserves excitation-number covariance. Because the
zero- and two-excitation subspaces are one dimensional, only their populations
can appear. The one-excitation subspace additionally permits the coherence
between $\ket{2}$ and $\ket{3}$~\cite{Wang2019,Eastham2016}. Full
secularization discards the cross-frequency terms that couple these two
eigenstates. For the nondegenerate spectrum considered here, the secular
stationary state is therefore diagonal. The corresponding
populations must be
recomputed from the secular generator rather than obtained by simply deleting
$c$ from Eq.~\eqref{eq:ss_block}.

The surviving NESS coherence has the analytical expression
\begin{equation}
c=\frac{S}{N(2A-i\Omega/J)}.
\label{eq:rho23_complex}
\end{equation}
Here $N$ is the normalization factor
\begin{equation}
N=(1+2\tilde n_1)(1+2\tilde n_2)
-\frac{4A^2\delta n_1\delta n_2}
{A^2+\left(\Omega/2J\right)^2}.
\label{eq:N_definition}
\end{equation}
The remaining parameters $A$ and $S$ are
\begin{subequations}
\begin{align}
A&=1+\tilde n_1+\tilde n_2,
\label{eq:A_definition}\\
S&=\delta n_1(1+2\tilde n_2)+\delta n_2(1+2\tilde n_1).
\label{eq:Sb}
\end{align}
\end{subequations}
All three quantities $N$, $A$, and $S$ are dimensionless. The effective
occupations and bath-contrast terms entering them are
\begin{subequations}
\begin{align}
&\tilde n_1=\cos^2\left(\frac{\theta}{2}\right)n_1(\omega'_1)
+\sin^2\left(\frac{\theta}{2}\right)n_2(\omega'_1),\label{eq:ntilde1}\\
&\tilde n_2=\sin^2\left(\frac{\theta}{2}\right)n_1(\omega'_2)
+\cos^2\left(\frac{\theta}{2}\right)n_2(\omega'_2),\label{eq:ntilde2}\\
&\delta n_j=\frac{\sin\theta}{2}
 (n_1(\omega'_j)-n_2(\omega'_j)).
\label{eq:delta_n_main}
\end{align}
\end{subequations}
Here $j=1,2$. The effective occupations $\tilde n_j$ are mixing-weighted bath
populations at the two frequencies $\omega_1'$ and $\omega_2'$, whereas
$\delta n_j$ quantify the bath contrast that drives interference between the
two transition pathways. Thus $S$ is the nonequilibrium source of $c$, and it
requires both temperature bias and eigenstate mixing.

The complete expressions for the four populations are given in
Appendix~\ref{app:br_steady_state}. For the ergotropy analysis in
Sec.~\ref{sec:steady_ergotropy}, the key population combination is the
imbalance $\Delta p=p_3-p_2$, whose exact expression is
\begin{equation}
\Delta p=\frac{\tilde n_2-\tilde n_1
+K((\delta n_2)^2-(\delta n_1)^2)}{N},
\label{eq:Delta_p_analytic}
\end{equation}
with
\begin{equation}
K=\frac{A}{A^2+\left(\Omega/2J\right)^2}.
\label{eq:K_definition}
\end{equation}
The factor $K$ is also dimensionless.
The difference $\tilde n_2-\tilde n_1$ is the direct population bias between
the two transition frequencies. The term proportional to $K$ is the population
correction generated by the cross-frequency dynamics. Together, $\Delta p$
and $|c|$ set the eigenvalue separation in the one-excitation block and are
therefore the physical variables needed for the ergotropy analysis.

This analytical construction follows Refs.~\cite{Wang2019,Zhang2020}. We use
it only at parameter points for which $N>0$ and $\rho^{\rm ss}$ is positive
semidefinite. As a consistency check,
when $T_1=T_2=T$, all $\delta n_j$ vanish, so $c=0$ and
Eq.~\eqref{eq:ss_block} reduces to the Gibbs
state $e^{-H_S/T}/\operatorname{Tr}(e^{-H_S/T})$. When $\lambda=0$ at nonzero
detuning, $\sin\theta=0$ also gives $c=0$, but unequal reservoirs produce
the product of the two local thermal states.

\section{Ergotropy of the nonequilibrium steady state}
\label{sec:steady_ergotropy}

In this section, we introduce ergotropy and its decomposition in
Subsection~\ref{subsec:ergotropy_def}. We then use the structure of the NESS in
Eq.~\eqref{eq:ss_block} to analyze its coherent and incoherent contributions in
Subsections~\ref{subsec:coherence_activation} and~\ref{subsec:coh_incoh},
respectively.

\subsection{Ergotropy and coherent--incoherent decomposition}
\label{subsec:ergotropy_def}

For a state $\rho$ of a system with Hamiltonian $H_S$, ergotropy is the maximum
work extractable through a cyclic unitary process~\cite{Allahverdyan2004}:
\begin{equation}
\mathcal E(\rho,H_S)=\Tr(\rho H_S)-\min_U\Tr(U\rho U^\dagger H_S).
\label{eq:ergotropy_def}
\end{equation}
Here the minimization is over all global unitaries on the system, and cyclicity
means that the Hamiltonian returns to $H_S$ at the end of the process. To
evaluate the minimum, write
$H_S=\sum_\alpha E_\alpha\ket{\alpha}\bra{\alpha}$ with
$E_1\le E_2\le\cdots$, and order the eigenvalues of $\rho$ as
$r_1^\downarrow\ge r_2^\downarrow\ge\cdots$. Since a unitary preserves these
eigenvalues, the minimum energy is attained by the passive state
$\pi_\rho=\sum_\alpha r_\alpha^\downarrow\ket{\alpha}\bra{\alpha}$, which
pairs the largest eigenvalue of the state with the lowest energy, the next largest
with the next-lowest energy, and so on. A state already arranged in this way is
passive and has zero ergotropy~\cite{PuszWoronowicz1978,Lenard1978,Allahverdyan2004}.
Therefore,
\begin{equation}
\mathcal E(\rho,H_S)=\Tr(\rho H_S)-\Tr(\pi_\rho H_S).
\end{equation}
Thus calculating ergotropy requires diagonalizing $\rho$, sorting its
eigenvalues in decreasing order, and pairing them with the energy eigenstates
of $H_S$ in order of increasing energy.

Ergotropy can originate from nonpassive energy populations and from coherence in the energy eigenbasis, motivating its decomposition into incoherent and coherent contributions~\cite{Francica2020,ChinPhysLett.43.060602}. For nondegenerate $H_S$, define energy dephasing as $\mathcal D_H(\rho)=\sum_\alpha\ket{\alpha}\bra{\alpha}\rho
\ket{\alpha}\bra{\alpha}$. This map preserves the mean energy, and the incoherent ergotropy is defined as~\cite{Lobejko2022}
\begin{equation}
\label{eq:general_incoherent_definition}    
\mathcal E_{\rm inc}=\mathcal E(\mathcal D_H(\rho),H_S).
\end{equation}
The remainder is the coherent ergotropy, namely
\begin{equation}
\label{eq:Ecoh_def}
\mathcal E_{\rm coh}=\mathcal E(\rho,H_S)-\mathcal E_{\rm inc}.
\end{equation}
Energy dephasing leaves both the energy populations and the mean energy
unchanged, but it removes the off-diagonal elements that alter the spectrum.
Thus $\mathcal E_{\rm inc}$ measures activity already visible in
the energy populations, whereas $\mathcal E_{\rm coh}$ is the additional work
made available because coherence changes the eigenvalues entering the passive
rearrangement.

\subsection{Coherent ergotropy activation}
\label{subsec:coherence_activation}

The structure of the NESS in Eq.~\eqref{eq:ss_block} permits an analytical
evaluation of the NESS ergotropy. In particular, it gives the following
result for the activation of coherent ergotropy.

\begin{theorem}
\label{thm:coherence_activation}
Consider any physical NESS of the analytical stationary solution in
Eq.~\eqref{eq:ss_block} with system Hamiltonian $H_S$
in Eq.~\eqref{eq:HS}. Its coherent ergotropy is
\begin{equation}
\mathcal E_{\rm coh}=\frac{\Omega}{2}
(\sqrt{(\Delta p)^2+4|c|^2}-|\Delta p|),
\label{eq:Ecoh_raw}
\end{equation}
and is nonzero if and only if $c\ne0$.
\end{theorem}

\begin{proof}
The block of Eq.~\eqref{eq:ss_block} on the coherent subspace spanned by
$\{\ket2,\ket3\}$ has eigenvalues
\begin{equation}
r_\pm=\frac{p_2+p_3}{2}\pm\frac12\sqrt{(\Delta p)^2+4|c|^2}.
\label{eq:rpm}
\end{equation}
For the analytical NESS of the present model, positivity of the stationary
state together with the exact stationary relations implies the ordering
\begin{equation}
p_1\ge r_+\ge r_-\ge p_4.
\label{eq:automatic_eigenvalue_ordering}
\end{equation}
The complete analytical proof is given in
Appendix~\ref{app:automatic_ordering}.
Moreover, $r_+\ge\max(p_2,p_3)$ and $r_-\le\min(p_2,p_3)$. Thus $p_1$
remains paired with $E_1$ and $p_4$ with $E_4$, while only the two middle
populations can be reordered. Because dephasing also preserves the
mean energy, Eqs.~\eqref{eq:general_incoherent_definition}
and~\eqref{eq:Ecoh_def} give
\begin{multline}
\mathcal E_{\rm coh}
=E_2(\max(p_2,p_3)-r_+)
  +E_3(\min(p_2,p_3)-r_-)\\
=\frac{\Omega}{2}
  (\sqrt{(\Delta p)^2+4|c|^2}-|\Delta p|).
\end{multline}
This proves
Eq.~\eqref{eq:Ecoh_raw}. Since $\Omega>0$ and
$\sqrt{(\Delta p)^2+4|c|^2}>|\Delta p|$ exactly when $c\ne0$,
$\mathcal E_{\rm coh}>0$ if and only if $c\ne0$.
\end{proof}

Because $\mathcal E=\mathcal E_{\rm inc}+\mathcal E_{\rm coh}$ with
$\mathcal E_{\rm inc}\ge0$, the theorem also implies
$\mathcal E>0$ whenever $c\ne0$.
For a physical state with strictly ordered energies, coherence-only
activation is equivalent to $p_1\ge p_2\ge p_3\ge p_4$ and $c\ne0$.
Population ordering makes the dephased state passive, whereas the presence of
coherence makes positive work extractable from the NESS. The passive
rearrangement is attained by a global unitary that diagonalizes the
single-excitation subspace, assigns $r_+$ and $r_-$ to $E_2$ and $E_3$,
respectively, and leaves the outer energy levels fixed.

\subsection{Incoherent and total ergotropy}
\label{subsec:coh_incoh}

The NESS structure in Eq.~\eqref{eq:ss_block} also gives a compact expression for the incoherent ergotropy.

\begin{theorem}
\label{thm:incoherent_ergotropy}
Consider any physical NESS of the analytical stationary solution in
Eq.~\eqref{eq:ss_block} with system Hamiltonian $H_S$
in Eq.~\eqref{eq:HS}. Its incoherent ergotropy is
\begin{equation}
\mathcal E_{\rm inc}=\Omega\max(0,\Delta p).
\label{eq:E_inc_compact}
\end{equation}
Equivalently, the dephased NESS is passive if and only if $\Delta p\le0$.
\end{theorem}

\begin{proof}
Energy dephasing removes $c$ while leaving the populations unchanged, so
$\mathcal D_H(\rho^{\rm ss})=\operatorname{diag}(p_1,p_2,p_3,p_4)$. From
$p_1\ge r_+$, $r_-\ge p_4$, $r_+\ge\max(p_2,p_3)$, and
$r_-\le\min(p_2,p_3)$, one obtains
\begin{equation}
p_1\ge\max(p_2,p_3)\ge\min(p_2,p_3)\ge p_4.
\end{equation}
Thus $p_1$ and $p_4$ already have their passive assignments, and only the two
middle populations may need to be reordered. If $\Delta p=p_3-p_2\le0$, then
$p_2\ge p_3$ and the dephased state is passive, giving
$\mathcal E_{\rm inc}=0$. If $\Delta p>0$, its passive state is obtained by
exchanging $p_2$ and $p_3$, and hence
\begin{equation}
\mathcal E_{\rm inc}
=(p_2E_2+p_3E_3)-(p_3E_2+p_2E_3)=\Omega\Delta p.
\end{equation}
Combining the two cases proves Eq.~\eqref{eq:E_inc_compact} and the stated
passivity criterion.
\end{proof}

Combining Theorem~\ref{thm:incoherent_ergotropy} with
Theorem~\ref{thm:coherence_activation} gives the total ergotropy
\begin{equation}
\mathcal E=\frac{\Omega}{2}
(\sqrt{(\Delta p)^2+4|c|^2}+\Delta p).
\label{eq:E_total_compact}
\end{equation}
At fixed $\Omega$ and $|c|$, reducing $|\Delta p|$ increases the coherent
ergotropy. For nonzero $|c|$ and
$\Delta p\le0$, the ergotropy is coherence-only. For
$\Delta p>0$ and $|c|>0$, population and coherence contributions coexist.
For $c=0$, activity instead depends on population ordering, and $\Delta p=0$ is a population-ordering threshold. At fixed $H_S$ and populations, the small-coherence response is
quadratic for $\Delta p\ne0$:
\begin{equation}
\mathcal E_{\rm coh}=\frac{\Omega|c|^2}{|\Delta p|}
-\frac{\Omega|c|^4}{|\Delta p|^3}
+O\!\left(\frac{\Omega|c|^6}{|\Delta p|^5}\right),
\label{eq:fixed_population_response}
\end{equation}
for $|c|\ll|\Delta p|$, whereas
$\mathcal E_{\rm coh}=\Omega|c|$ exactly at $\Delta p=0$.

The preceding analysis identifies $c$ and $\Delta p$ as the two state-level
quantities controlling distinct mechanisms of ergotropy activation: coherence
and population ordering, respectively. This characterization, however, does
not yet reveal how the physical parameters of the model, such as the reservoir
temperatures, qubit detuning, and exchange coupling, determine the NESS ergotropy. We therefore turn to a systematic analysis of the
dependence of NESS ergotropy on these physical parameters, beginning with the
resonant case $\Delta\omega=0$.

\section{Resonant case: coherence-only ergotropy}
\label{sec:resonance}

In this section, we first show in Sec.~\ref{subsec:resonant_coherence_only} that the NESS ergotropy of two resonant qubits with $\Delta\omega=0$ is entirely coherent. We then discuss the influence of temperature bias and exchange coupling strength on NESS ergotropy at resonance in Secs.~\ref{subsec:resonant_bias} and~\ref{subsec:resonant_exchange}, respectively.

\subsection{Coherence-only ergotropy activation}
\label{subsec:resonant_coherence_only}

\begin{figure*}[t]
\centering
\includegraphics[width=0.98\textwidth]{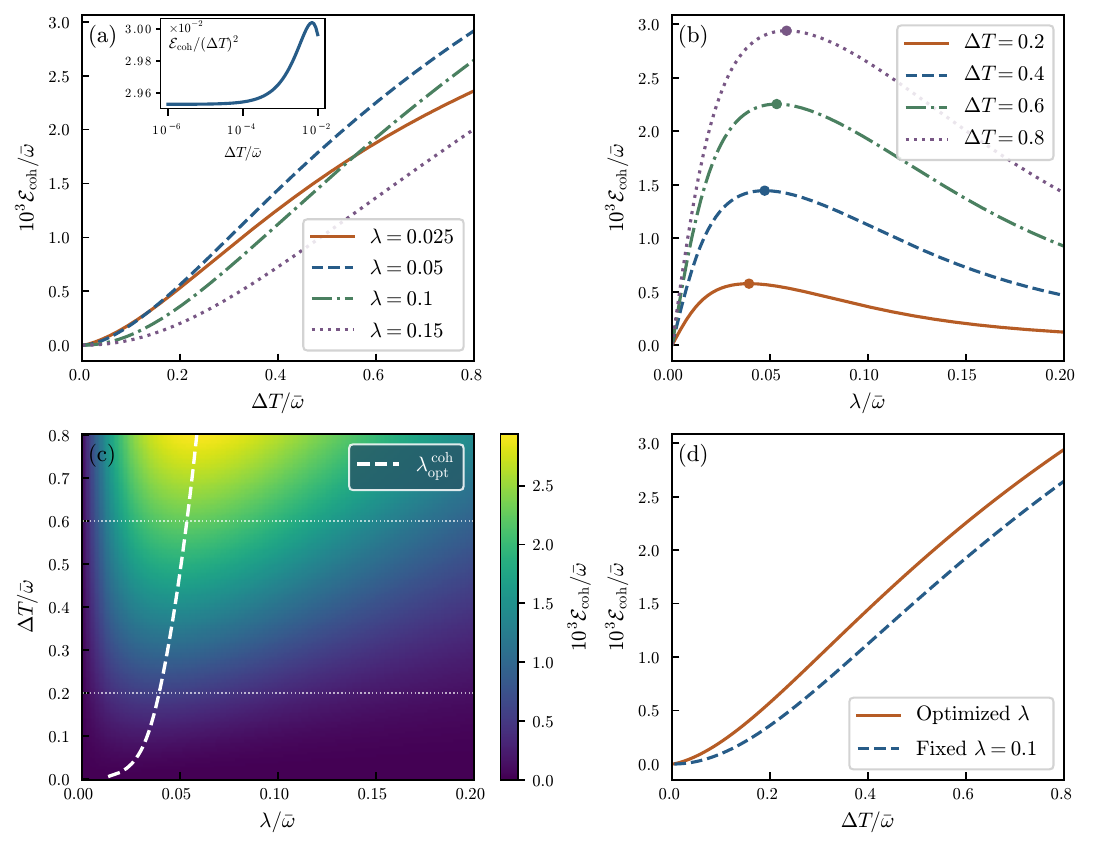}
\caption{Coherence-only ergotropy of the NESS at resonance ($\Delta\omega=0$). All panels use the parameters
$T_1/\bar\omega=0.2$, $T_2=T_1+\Delta T$, and $J/\bar\omega=0.01$. (a) Coherent ergotropy of the NESS as a function of the temperature bias $\Delta T$ for different exchange couplings $\lambda$. The inset shows $\mathcal E_{\rm coh}/(\Delta T)^2$ near equilibrium with $\lambda/\bar\omega=0.05$. (b) Coherent ergotropy of the NESS as a function of the exchange coupling $\lambda$ for different temperature biases $\Delta T$.
Filled dots mark independently optimized maxima. (c) Coherent ergotropy in the $(\lambda,\Delta T)$ plane. The white dashed curve marks
the optimized value $\lambda_{\rm opt}^{\rm coh}$, and the white dotted
horizontal lines indicate $\Delta T/\bar\omega=0.2,0.6$.
(d) Coherent ergotropy as a function of the temperature bias $\Delta T$ with optimized coupling
$\lambda=\lambda_{\rm opt}^{\rm coh}$ (solid) and fixed
$\lambda/\bar\omega=0.1$ (dashed).}
\label{fig:resonance}
\end{figure*}

We refer to the case of two qubits with the same bare frequencies, $\Delta\omega=0$, as the resonant case. Correspondingly, the mixing angle $\theta$ defined in Eq.~\eqref{eq:theta_definition} equals $\pi/2$, which gives two maximally entangled eigenstates of the model. Furthermore, the incoherent ergotropy vanishes and only coherent ergotropy is present in the NESS.

\begin{theorem}
\label{thm:resonant_coherence_only}
Consider any physical NESS with $\Delta\omega=0$ and
$0<\lambda<2\bar\omega$. Then $\Delta p<0$ and the incoherent ergotropy vanishes, giving $\mathcal E_{\rm inc}=0$ and $\mathcal E=\mathcal E_{\rm coh}$. 
\end{theorem}

\begin{proof}
At resonance, $\theta=\pi/2$, $\Omega=\lambda$, and the effective occupation parameters in Eqs.~\eqref{eq:ntilde1}--\eqref{eq:delta_n_main} reduce to
\begin{equation}
\tilde n_j=\frac{n_1(\omega'_j)+n_2(\omega'_j)}{2},\qquad \delta n_j=\frac{n_1(\omega'_j)-n_2(\omega'_j)}{2},
\label{eq:resonant_effective_occupations}
\end{equation}
with $j=1,2$. Since the transition frequencies satisfy $\omega'_1<\omega'_2$ and every Bose occupation decreases with
frequency, $\tilde n_1>\tilde n_2$.

For a nonequilibrium environment with unequal temperatures, the function
$|n_1(\omega)-n_2(\omega)|$ is also strictly decreasing for $\omega>0$.
Indeed, for $T_{\rm h}>T_{\rm c}$, its derivative is strictly negative:
\begin{equation}
\frac{d}{d\omega}\left(n(\omega,T_{\rm h})-n(\omega,T_{\rm c})\right)<0.
\label{eq:bose_difference_monotonicity}
\end{equation}
Hence
$|\delta n_1|>|\delta n_2|$ when $T_1\ne T_2$, while both vanish when
$T_1=T_2$. Since the parameters $K$ and $N$ given by Eqs.~\eqref{eq:K_definition} and~\eqref{eq:N_definition}, respectively, are both positive, Eq.~\eqref{eq:Delta_p_analytic} therefore
gives
\begin{equation}
\Delta p=\frac{\tilde n_2-\tilde n_1
+K\left((\delta n_2)^2-(\delta n_1)^2\right)}{N}<0.
\label{eq:resonant_population_ordering}
\end{equation}
Theorem~\ref{thm:incoherent_ergotropy} then yields
$\mathcal E_{\rm inc}=0$. If $T_1\ne T_2$, $\delta n_1$ and $\delta n_2$
have the same nonzero sign, so Eq.~\eqref{eq:Sb} gives $S\ne0$ and
Eq.~\eqref{eq:rho23_complex} gives $c\ne0$.
Theorem~\ref{thm:coherence_activation} then implies
$\mathcal E=\mathcal E_{\rm coh}>0$. At equilibrium, the Gibbs
state is passive with vanishing ergotropy.
\end{proof}

At resonance, the NESS coherence $c\ne0$
requires retaining the cross-frequency dissipative terms beyond the full
secular approximation~\cite{Wang2019,Eastham2016}. Full secularization removes
these terms, yielding $c^{\rm sec}=0$ and a passive stationary state with
$\mathcal E^{\rm sec}=0$. Thus, the secular and nonsecular descriptions differ
qualitatively in the resonant case. A more
comprehensive comparison between the partial and full secular approximations, including the changes in stationary
populations, is presented in Sec.~\ref{sec:secular_comparison}.

\subsection{Temperature-bias response}
\label{subsec:resonant_bias}
Nonequilibrium environments have been shown to generate and enhance energy-basis
coherence in the NESS of three-level systems~\cite{LiCaiSun2015} and
coupled qubits~\cite{Wang2019}. These results suggest that increasing the temperature bias can enhance coherent ergotropy in the resonant case. Note that the coherent ergotropy given by Eq.~\eqref{eq:Ecoh_raw} depends on both the NESS coherence $c$ and the population imbalance $\Delta p$. We examine the activation and enhancement of ergotropy via the nonequilibrium environment by fixing the cold bath at $T_1=T_0$ and setting $T_2=T_0+\Delta T$, with
$\Delta T\ge0$.
Figure~\ref{fig:resonance}(a) shows that coherent ergotropy grows from zero
at equilibrium as the temperature bias $\Delta T$ increases for each fixed exchange
coupling $\lambda$. Its magnitude and growth depend on the exchange coupling $\lambda$, and
the relative ordering of the curves can change as the temperature bias increases.

Near equilibrium, the ergotropy activation can be obtained analytically. For fixed
nondegenerate $H_S$ and $T_0>0$, the coherence magnitude $|c|$ is
linear in the small positive temperature bias. With a finite passive population gap $\Delta p<0$ at
equilibrium, the small-coherence response in Eq.~\eqref{eq:fixed_population_response}
gives $\mathcal E_{\rm coh}\sim\Omega|c|^2/|\Delta p|$, yielding the quadratic law
\begin{equation}
\mathcal E_{\rm coh}=C_T(\Delta T)^2+O((\Delta T)^3).
\label{eq:temperature_bias_law}
\end{equation}
The positive coefficient $C_T$ is evaluated at fixed $H_S,J,T_0$ and has
units of inverse energy. Appendix~\ref{app:asymptotics} gives its
explicit expression. The inset of Fig.~\ref{fig:resonance}(a) shows the exact ratio
$\mathcal E_{\rm coh}/(\Delta T)^2$ in units of $10^{-2}/\bar\omega$.
The curve approaches $C_T$ as $\Delta T\to0^+$, consistent with
Eq.~\eqref{eq:temperature_bias_law}. At the inset parameters
$\lambda/\bar\omega=0.05$ and $T_0/\bar\omega=0.2$, the limiting value is
$\bar\omega C_T\simeq2.9529\times10^{-2}$.
This expansion assumes
$|\Delta T|/T_0\ll1$ at fixed exchange coupling. 

The fixed-cold-bath setup raises the mean temperature together with the temperature bias.
At resonance, swapping the bath temperatures $T_1$ and $T_2$ while keeping
all other parameters fixed leaves the NESS ergotropy unchanged:
$\mathcal E(T_1,T_2)=\mathcal E(T_2,T_1)$.
Reversing $\Delta T=T_2-T_1$ at fixed $T_1$ does not swap the bath pair,
and the ergotropy need not be even along this path. A cubic correction in
Eq.~\eqref{eq:temperature_bias_law} is therefore allowed.
The reported enhancement refers to the
finite interval of temperature bias in Fig.~\ref{fig:resonance}(a). The dependence of the temperature-bias enhancement on exchange coupling motivates the
analysis of exchange-coupling optimization in Sec.~\ref{subsec:resonant_exchange}.

\subsection{Exchange-coupling dependence and optimization}
\label{subsec:resonant_exchange}

At resonance $\Delta\omega=0$, varying the exchange coupling $\lambda$
affects the NESS coherence $c$ in two ways~\cite{Wang2019}.
Increasing the energy separation $\Omega=\lambda$ between $\ket2$ and
$\ket3$ suppresses $|c|$ when the bath occupations are held fixed, as shown in Eq.~\eqref{eq:rho23_complex}. Changing $\lambda$ also shifts the
transition frequencies $\omega'_{1,2}=\bar\omega\mp\lambda/2$ and thus
changes the bath occupations, which determine the coherence source $S$,
normalization $N$, and coherence linewidth $\gamma_{\rm eff}=2AJ$.
Analytically, we find that the coherent ergotropy grows linearly with sufficiently weak
exchange coupling at fixed temperature bias $\Delta T$. Setting
$\Delta\omega=0$ before taking $\lambda\to0^+$, with $J$ and
$\bar\omega$ fixed, gives
\begin{equation}
\mathcal E_{\rm coh}=
\frac{|n_1-n_2|}{2(1+2n_1)(1+2n_2)}\lambda+O(\lambda^2),
\label{eq:resonant_weak_exchange}
\end{equation}
where $n_j=n(\bar\omega,T_j)$. The derivation is given in Appendix~\ref{app:asymptotics}. The linear coefficient vanishes at equal bath temperatures, where the
NESS has zero ergotropy for every exchange coupling. At exactly
$\Delta\omega=\lambda=0$, the first and second excited states are degenerate, so coherence within
this degenerate subspace requires a separate treatment of energy dephasing.

At finite exchange coupling, coherent ergotropy becomes nonmonotonic. See Fig.~\ref{fig:resonance}(b). Equation~\eqref{eq:Ecoh_raw} shows that the ergotropy scale grows as
$\Omega=\lambda$, while both the stationary coherence $|c|$ and the
passive population gap $|\Delta p|$ vary with exchange coupling $\lambda$. Increasing the energy splitting $\Omega$ relative to the coherence linewidth
$\gamma_{\rm eff}=2AJ$ suppresses cross-frequency interference in
Eq.~\eqref{eq:rho23_complex}. Changing $\lambda$ also shifts both transition
frequencies, $\omega'_{1,2}=\bar\omega\mp\lambda/2$, and hence changes
the bath occupations, the source $S$, the normalization $N$, and the linewidth. The nonmonotonic ergotropy response results from these changes
together with the growing energy splitting.

We define the optimal exchange coupling $\lambda_{\rm opt}^{\rm coh}$ by
\begin{equation}
\mathcal E_{\rm coh}(\lambda_{\rm opt}^{\rm coh})
=\max_{0\le\lambda/\bar\omega<2}
\mathcal E_{\rm coh}(\lambda),
\label{eq:resonant_optimal_exchange}
\end{equation}
with $T_1,T_2,J,\bar\omega$ fixed and $\Delta\omega=0$. The maximization is restricted to positive semidefinite NESSs. The markers in Fig.~\ref{fig:resonance}(b) indicate
independently determined optimal couplings $\lambda_{\rm opt}^{\rm coh}$. Figure~\ref{fig:resonance}(c) shows the $(\lambda,\Delta T)$ plane, with a white dashed line marking the optimal coupling $\lambda_{\rm opt}^{\rm coh}$ at each temperature bias $\Delta T$. The optimal coupling $\lambda_{\rm opt}^{\rm coh}$ shifts toward larger exchange coupling as the temperature bias
increases. Figure~\ref{fig:resonance}(d) compares $\mathcal E_{\rm coh}$ evaluated at the optimal coupling
$\lambda_{\rm opt}^{\rm coh}$ with the fixed-coupling reference
$\lambda/\bar\omega=0.1$. The
optimized ergotropy grows with temperature bias and exceeds the reference at
nonzero temperature bias. Thus attaining the maximum coherent ergotropy requires
adjusting the exchange coupling with the temperature bias.

\section{Finite detuning: population activation and work composition}
\label{sec:detuning}

In this section, we first show that finite detuning with $\Delta\omega\neq 0$ can activate incoherent ergotropy and present an exact activation criterion in Sec.~\ref{subsec:population_activation}. We then discuss the competition between incoherent and coherent ergotropy in Sec.~\ref{subsec:detuning_exchange}. Finally, we study exchange-coupling optimization of ergotropy in Sec.~\ref{subsec:population_optimization}.

\begin{figure*}[t]
\centering
\includegraphics[width=0.99\textwidth]{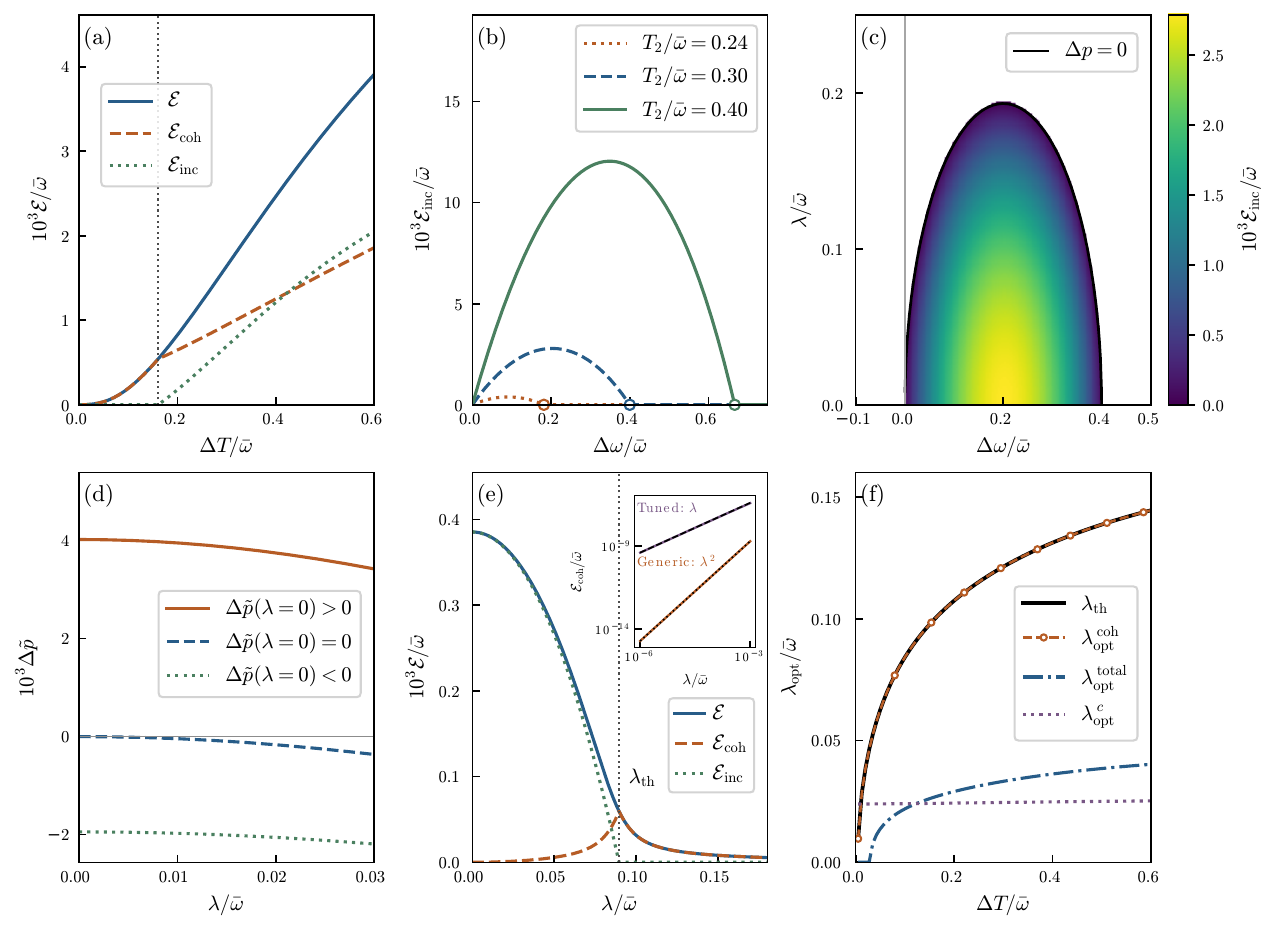}
\caption{Finite-detuning activation, composition, and optimization of
NESS ergotropy. All panels use $J/\bar\omega=0.01$ and
$T_1/\bar\omega=0.2$, with $T_2=T_1+\Delta T$.
(a) Total (blue solid), coherent (orange dashed), and incoherent (green dotted)
ergotropy versus temperature bias $\Delta T$ at $\Delta\omega/\bar\omega=0.02$ and
$\lambda/\bar\omega=0.1$. The gray vertical dotted line marks the onset of
incoherent ergotropy.
(b) Incoherent ergotropy of two uncoupled qubits ($\lambda=0$) versus detuning $\Delta\omega$. Open circles mark the activation endpoints
where $n(\omega_2,T_2)=n(\omega_1,T_1)$.
(c) Incoherent-ergotropy map at $T_2/\bar\omega=0.3$. The black solid curve is the exact
$\Delta p=0$ activation boundary. 
(d) Population imbalance $\Delta\tilde p=N\Delta p$ versus exchange coupling at
$\Delta\omega/\bar\omega=0.1$. Solid, dashed, and dotted curves correspond
to $\Delta\tilde p(\lambda=0)>0$, $=0$, and $<0$, respectively, with
$T_2/\bar\omega=0.24$, $0.2(1.05/0.95)$, and $0.21$.
(e) Ergotropy composition for the initially active case of (d), with
$T_2/\bar\omega=0.24$. The gray vertical dotted line marks
$\lambda_\text{th}$, where incoherent ergotropy vanishes and coherent ergotropy is maximal. The inset compares the generic quadratic
($\Delta\tilde p(\lambda=0)\ne0$) and tuned linear
($\Delta\tilde p(\lambda=0)=0$) coherent-ergotropy responses at
$T_2/\bar\omega=0.24$ and $0.2(1.05/0.95)$, respectively.
Black dotted and dashed lines show the corresponding leading-order laws
in Eqs.~\eqref{eq:detuned_generic_condition} and
\eqref{eq:detuned_tuned_coherent_law}.
(f) Incoherent-ergotropy threshold $\lambda_\text{th}$ (black solid) and the
independently determined optimal couplings
$\lambda_{\rm opt}^{\rm coh}$ (orange dashed with open circles),
$\lambda_{\rm opt}^{\rm total}$ (blue dash-dotted), and
$\lambda_{\rm opt}^{c}$ (purple dotted), defined in
Eq.~\eqref{eq:distinct_exchange_optimizers}, versus temperature bias at
$\Delta\omega/\bar\omega=0.02$. }
\label{fig:detuning}
\end{figure*}

\subsection{Incoherent ergotropy activation criterion}
\label{subsec:population_activation}

In contrast to the resonant case $\Delta\omega=0$, finite detuning allows the single-excitation populations to reverse their ordering. The population imbalance $\Delta p$ can therefore become positive, activating incoherent ergotropy. Figure~\ref{fig:detuning}(a) shows the numerically evaluated coherent,
incoherent, and total ergotropy in terms of the temperature bias $\Delta T$. Below a threshold, ergotropy is entirely coherent. The analytical NESS presented in Sec.~\ref{subsec:compact_param} allows us to determine the exact activation threshold for incoherent ergotropy.

With symmetric flat reservoir spectra $J_1=J_2=J$, the ergotropy has the symmetry 
\begin{equation}
\mathcal E_\alpha(\Delta\omega,T_1,T_2)
=\mathcal E_\alpha(-\Delta\omega,T_2,T_1),
\label{eq:bias_detuning_symmetry}
\end{equation}
for $\alpha\in\{{\rm coh},{\rm inc},{\rm total}\}$. We therefore take
$T_2>T_1>0$ without loss of generality. Note that the sign of $\Delta\omega$ specifies which reservoir is attached to the higher-frequency bare qubit.
At the transition frequencies $\omega'_1<\omega'_2$ in
Eq.~\eqref{eq:transition_frequencies}, define the mean occupation as
\begin{equation}
\bar n_a=\frac{n(\omega'_a,T_1)+n(\omega'_a,T_2)}{2},
\end{equation}
with $a=1,2$. Since $T_2>T_1$, the bath-contrast terms $\delta n_a$
defined in Eq.~\eqref{eq:delta_n_main} are strictly negative for
$\lambda>0$, so $-(\delta n_1+\delta n_2)>0$.
For nonzero exchange and finite nonzero detuning, the exact activation
criterion for incoherent ergotropy is as follows.

\Needspace{9\baselineskip}
\begin{theorem}
\label{thm:inc-activation}
Consider any physical NESS with nonzero exchange $\lambda>0$. The incoherent ergotropy is activated if and only if
\begin{equation}
\cot\theta>
-\frac{\bar n_1-\bar n_2}{\delta n_1+\delta n_2}
+K(\delta n_2-\delta n_1).
\label{eq:exact-inc-activation}
\end{equation}
\end{theorem}

\begin{proof}
Theorem~\ref{thm:incoherent_ergotropy} gives
$\mathcal E_{\rm inc}=\Omega\max(0,\Delta p)$. Since $N>0$,
activation is equivalent to $\Delta\tilde p=N\Delta p>0$. Substituting
Eqs.~\eqref{eq:ntilde1}--\eqref{eq:delta_n_main} into
Eq.~\eqref{eq:Delta_p_analytic} yields the exact identity
\begin{multline}
\Delta\tilde p=-(\bar n_1-\bar n_2)
-\cot\theta(\delta n_1+\delta n_2)\\
-K((\delta n_1)^2-(\delta n_2)^2).
\label{eq:unnormalized-imbalance-direct}
\end{multline}
Factorizing $(\delta n_1)^2-(\delta n_2)^2
=(\delta n_1+\delta n_2)(\delta n_1-\delta n_2)$ and dividing
$\Delta\tilde p>0$ by $-(\delta n_1+\delta n_2)>0$ gives precisely
Eq.~\eqref{eq:exact-inc-activation}.
\end{proof}

The criterion in Eq.~\eqref{eq:exact-inc-activation} is determined by
the Hamiltonian parameters $\bar\omega$, $\Delta\omega$, and $\lambda$, and the environmental parameters $T_1,T_2$, and $J$.
The criterion is exact within the partial-secular Bloch--Redfield
description. Note that the second term on the right of Eq.~\eqref{eq:exact-inc-activation} is the nonsecular population correction from
Eq.~\eqref{eq:Delta_p_analytic}. Incoherent work
appears only when the ratio of detuning to exchange strength,
$\cot\theta=\Delta\omega/\lambda$, exceeds the threshold. Physically, it also imposes a directional relation between detuning $\Delta\omega$ and temperature bias $\Delta T$, as stated in the following corollary.

\begin{corollary}
\label{cor:activation_direction}
Incoherent ergotropy is activated only when the hotter reservoir is attached to the
higher-frequency bare qubit, i.e., $\Delta\omega\Delta T>0$.
\end{corollary}

\begin{proof}
The Bose occupation $n(\omega,T)$ decreases strictly with frequency at fixed
temperature, so $\bar n_1>\bar n_2$.
The magnitude of its frequency derivative satisfies
\begin{equation}
-\partial_\omega n(\omega,T)
=\frac{1}{4T\sinh^2(\omega/(2T))},
\label{eq:bose_frequency_slope}
\end{equation}
and therefore increases with $T$ at fixed $\omega>0$: its logarithmic
temperature derivative is
$(x\coth(x/2)-1)/T>0$, with $x=\omega/T>0$.
Hence the positive occupation difference $n(\omega,T_2)-n(\omega,T_1)$
decreases strictly with frequency, and Eq.~\eqref{eq:delta_n_main} gives
$\delta n_1<\delta n_2<0$. Together with $\bar n_1>\bar n_2$ and $K>0$,
this makes the right-hand side of Eq.~\eqref{eq:exact-inc-activation}
strictly positive. Activation therefore requires
$\cot\theta=\Delta\omega/\lambda>0$, which implies $\Delta\omega>0$.
\end{proof}

For $\lambda=0$, $\Delta\omega>0$, and $T_2>T_1>0$, the NESS is the
product of the two local thermal states, with $\theta=0$,
$\Omega=\Delta\omega$, and vanishing coherence $c=0$. The effective occupations reduce to
$\tilde n_j=n(\omega_j,T_j)$, while $\delta n_1=\delta n_2=0$.
Equation~\eqref{eq:Delta_p_analytic} therefore gives
\begin{equation}
\Delta\tilde p(\lambda=0)=n(\omega_2,T_2)-n(\omega_1,T_1).
\label{eq:tilde_p_lambda_0}
\end{equation}
Incoherent ergotropy
is activated if and only if $n(\omega_2,T_2)>n(\omega_1,T_1)$. Since the Bose occupation decreases strictly with $\omega/T$, this is
equivalent to $T_2/T_1>\omega_2/\omega_1$.
The product-state population imbalance is
\begin{equation}
\Delta p(\lambda=0)=\frac12\left(
\tanh\frac{\omega_1}{2T_1}-\tanh\frac{\omega_2}{2T_2}\right).
\label{eq:uncoupled_population_imbalance}
\end{equation}
Using $\mathcal E_{\rm inc}=\Delta\omega\max(0,\Delta p)$ gives the
exact incoherent ergotropy shown in Fig.~\ref{fig:detuning}(b):
\begin{equation}
\mathcal E_{\rm inc}(\lambda=0)=\max\left(\frac{\Delta\omega}{2}\left(
\tanh\frac{\omega_1}{2T_1}-\tanh\frac{\omega_2}{2T_2}\right),0\right).
\label{eq:uncoupled_population_work}
\end{equation}
At fixed $T_1$ and bare frequencies $\omega_j$, increasing the temperature bias
$\Delta T$ promotes the activation of incoherent ergotropy in the uncoupled
limit and increases its magnitude once activated.

The directional restriction concerns incoherent ergotropy only. The opposite detuning orientation can still support coherent ergotropy because
the NESS coherence need not vanish. Figure~\ref{fig:detuning}(c) illustrates
the directional incoherent ergotropy activation region and its exact $\Delta p=0$ boundary.

\subsection{Competition between incoherent and coherent ergotropy}
\label{subsec:detuning_exchange}

Finite detuning $\Delta\omega>0$ allows both incoherent and coherent ergotropy to be activated. As the numerical results in Fig.~\ref{fig:detuning}(c) show, increasing the exchange coupling $\lambda$ decreases incoherent ergotropy. Recall that the incoherent ergotropy is activated when the population imbalance $\Delta p$ becomes positive, while the coherent ergotropy is activated when the coherence $|c|$ becomes nonzero. The competition between these two contributions can lead to nontrivial behavior in the total ergotropy as a function of the exchange coupling $\lambda$.

Fix $\Delta\omega>0$ and $T_2>T_1$, with $\lambda/\Delta\omega\ll1$.
Using the uncoupled population imbalance in Eq.~\eqref{eq:tilde_p_lambda_0},
the exact NESS solution gives
\begin{equation}
\label{eq:weak_exchange_population_order}
\Delta\tilde p(\lambda)=\Delta\tilde p(\lambda=0)-C_{\Delta\tilde p}\lambda^2+O(\lambda^4),
\end{equation}
where the coefficient $C_{\Delta\tilde p}>0$ is derived in Appendix~\ref{app:finite_detuning_analysis}. Weak exchange therefore decreases $\Delta\tilde p$ to
quadratic order and pushes population ordering toward the passive side. See Fig.~\ref{fig:detuning}(d) for a numerical illustration.  

The expansion in Eq.~\eqref{eq:weak_exchange_population_order} separates three regimes. If \mbox{$\Delta\tilde p(\lambda=0)>0$}, the uncoupled
state is population active. The population-work threshold $\lambda_\text{th}$ is defined by $\Delta p(\lambda_\text{th})=0$ and can be estimated in the weak-exchange regime as
\begin{equation}
\lambda_\text{th}\simeq\sqrt{\frac{\Delta\tilde p(\lambda=0)}{C_{\Delta\tilde p}}},
\end{equation}
below which the incoherent ergotropy is activated. If $\Delta\tilde p(\lambda=0)<0$, sufficiently weak exchange coupling cannot activate incoherent ergotropy. For either case with $\Delta\tilde p(\lambda=0)\ne0$, we find that coherent ergotropy responds quadratically to the exchange coupling:
\begin{equation}
\mathcal E_{\rm coh}=C_{\Delta \tilde p\neq 0}\lambda^2+O(\lambda^4),
\label{eq:detuned_generic_condition}
\end{equation}
with the coefficient given in Appendix~\ref{app:finite_detuning_analysis}.

At the tuned balance $\Delta\tilde p(\lambda=0)=0$, equivalently
$T_2/T_1=\omega_2/\omega_1$, the expansion in Eq.~\eqref{eq:weak_exchange_population_order}
instead gives $\Delta\tilde p<0$ for sufficiently small nonzero exchange. Thus
$\mathcal E_{\rm inc}=0$ in an interval near $\lambda=0$, while
$\Delta p=O(\lambda^2)$ and $|c|=O(\lambda)$ yield a linear response
\begin{equation}
\mathcal E_{\rm coh}=C_{\Delta \tilde p= 0}\lambda+O(\lambda^2),
\label{eq:detuned_tuned_coherent_law}
\end{equation}
where the coefficient $C_{\Delta \tilde p= 0}>0$ is derived in Appendix~\ref{app:finite_detuning_analysis}. 
At the balance point $\Delta \tilde p= 0$, the two eigenvalues of the single-excitation subspace of the uncoupled density matrix coincide, so the coherence $|c|=O(\lambda)$ splits them to first order, whereas $\Delta p=O(\lambda^2)$ contributes only a subleading correction.

Figure~\ref{fig:detuning}(e) follows the population-active case of
panel (d). Increasing exchange reduces incoherent work and generates
coherent work, which reaches its maximum at $\lambda_\text{th}$.
Beyond this threshold the work is entirely coherent, while total work
decreases along this coupling scan. The inset verifies the generic
quadratic and tuned linear weak-exchange laws in
Eqs.~\eqref{eq:detuned_generic_condition} and
\eqref{eq:detuned_tuned_coherent_law}.

\subsection{Exchange-coupling optimization}
\label{subsec:population_optimization}

At resonance $\Delta\omega=0$, Sec.~\ref{subsec:resonant_exchange} shows that coherent and
total ergotropy share an optimal exchange coupling, which shifts toward
larger values as the temperature bias increases. This optimum balances
the growing energy splitting against the suppression of steady-state
coherence. Finite detuning introduces incoherent ergotropy, so changing
exchange also changes the population contribution to the available work.
Maximizing coherence, coherent ergotropy, and total ergotropy therefore
requires distinguishing three optimal couplings:
\begin{subequations}
\label{eq:distinct_exchange_optimizers}
\begin{align}
&\lambda_{\rm opt}^{c}=\arg\max_\lambda|c(\lambda)|,\\
&\lambda_{\rm opt}^{\rm coh}=\arg\max_\lambda\mathcal E_{\rm coh}(\lambda),\\
&\lambda_{\rm opt}^{\rm total}=\arg\max_\lambda\mathcal E(\lambda).
\end{align}
\end{subequations}
The optimizations are performed with $\Delta\omega,T_1,T_2,J,$ and $\bar\omega$ fixed. 

Figure~\ref{fig:detuning}(f) compares the three optimal couplings as the
positive temperature bias increases. For the parameters shown, the coherent-ergotropy optimum coincides
with $\lambda_\text{th}$, and both
move toward larger exchange. Note that $\lambda_\text{th}$ is the exchange-coupling threshold for activation of incoherent ergotropy given by $\Delta p(\lambda_\text{th})=0$. Population balance removes the population mismatch $\Delta p$ that reduces coherent ergotropy in Eq.~\eqref{eq:Ecoh_raw}, so coherence alone contributes to ergotropy. The population-mismatch contribution changes the coherent-ergotropy slope from positive to negative at the crossing, placing the maximum at the threshold. 

The coherence optimum $\lambda_{\rm opt}^{c}$ varies only weakly with $\Delta T$ and appears nearly constant. Its coupling scale is set mainly by the competition between eigenstate mixing, which creates the coherence source, and the growing splitting $\Omega$ relative to the dissipative linewidth $\gamma_\text{eff} = 2AJ$, which suppresses cross-frequency interference. Increasing the temperature bias $\Delta T$ strengthens the coherence source while shifting this optimum only slightly. Consequently, maximizing the coherence amplitude does not locate the coupling that maximizes coherent ergotropy, which also depends on the energy splitting $\Omega$ and population imbalance $\Delta p$.

The total-ergotropy optimum has a different dependence on the temperature bias $\Delta T$. Within the
incoherent-ergotropy-active range, weaker temperature bias favors the uncoupled
limit $\lambda=0$: the coherent ergotropy generated by exchange cannot compensate for the
loss of incoherent ergotropy. Only above a finite temperature bias crossover does a
nonzero optimal exchange coupling emerge, after which it increases with the temperature bias.
This optimum of total ergotropy lies below $\lambda_\text{th}$ and retains an incoherent-ergotropy
contribution, balancing coherent ergotropy against the exchange-induced
changes in the populations and energy splitting. Thus a nonzero exchange
improves total ergotropy only in part of the temperature bias range, even though
it generates coherent ergotropy throughout.

\section{Comparison with the secular approximation}
\label{sec:secular_comparison}

\begin{figure*}[t]
\centering
\includegraphics[width=0.98\textwidth]{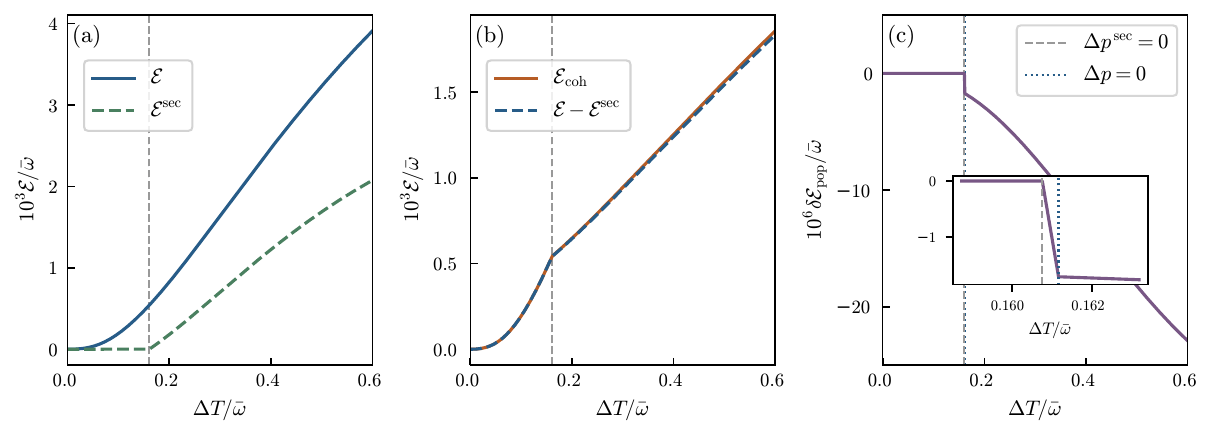}
\caption{Comparison of NESS ergotropy with and without the full secular approximation.
All panels use fixed $T_1/\bar\omega=0.2$ and $T_2=T_1+\Delta T$,
$\Delta\omega/\bar\omega=0.02$, $\lambda/\bar\omega=0.1$, and
$J/\bar\omega=0.01$.
(a) Total ergotropy $\mathcal E$ of the partially secular NESS (blue solid)
and $\mathcal E^{\rm sec}$ of the fully secular NESS (green dashed).
(b) Coherent ergotropy $\mathcal E_{\rm coh}$ (orange solid) and the total
ergotropy difference $\mathcal E-\mathcal E^{\rm sec}$ (blue dashed).
(c) Population correction $\delta\mathcal E_{\rm pop}=
\mathcal E_{\rm inc}-\mathcal E^{\rm sec}$ (purple solid). In all panels, gray dashed and blue dotted vertical lines mark
$\Delta p^{\rm sec}=0$ and $\Delta p=0$, respectively.}
\label{fig:secular_comparison}
\end{figure*}

In this section, we consider the NESS ergotropy obtained from the master equation with the full secular approximation~\cite{Davies1974,Breuer2002}. The full secular approximation removes the same-sign cross-frequency terms in the dissipator of Eq.~\eqref{eq:D0_Ds_split}. Specifically, the dissipator $\mathcal D_s$ including the same-sign cross-frequency terms is discarded. Consequently, the NESS in the energy basis is diagonal with vanishing coherence $c=0$. In the following, we use the superscript ``sec'' to denote the stationary state obtained under the full secular approximation. 

The outer population ordering, established in
Appendix~\ref{app:automatic_ordering}, also applies to the secular stationary state. Therefore, the ergotropy is entirely incoherent. Concretely, it is determined by
\begin{equation}
\mathcal E^{\rm sec}=\Omega\max(0,\Delta p^{\rm sec}),
\label{eq:E_sec_closed}
\end{equation}
where the population imbalance $\Delta p^{\rm sec}$ is
\begin{equation}
\Delta p^{\rm sec}=\frac{\tilde n_2-\tilde n_1}{N^{\rm sec}}.
\label{eq:secular_imbalance}
\end{equation}
Here the effective occupations $\tilde n_j$ are defined in Eqs. \eqref{eq:ntilde1} and \eqref{eq:ntilde2}, and the normalization factor $N^{\rm sec}$ is
\begin{equation}
N^{\rm sec}=(1+2\tilde n_1)(1+2\tilde n_2),
\label{eq:N_sec_definition}
\end{equation}
which can be obtained by removing the second term in $N$ defined in Eq.~\eqref{eq:N_definition}.

At resonance $\Delta\omega=0$, the lower transition has the larger effective occupation, so $\Delta p^{\rm sec}<0$ and the NESS state with the full secular approximation is always passive. In the detuned case $\Delta\omega\neq0$, the full secular approximation can activate incoherent ergotropy, but it is not equivalent to simply removing the coherent ergotropy of the NESS without the full secular approximation. This is because the populations are also affected by the secular approximation and $\mathcal D_H(\rho^{\rm ss})\ne\rho^{\rm sec}$ in general. Consequently, we have the difference
\begin{equation}
\mathcal E-\mathcal E^{\rm sec}=\mathcal E_{\rm coh}+\delta\mathcal E_{\rm pop},
\label{eq:accounting_identity}
\end{equation}
where $\delta\mathcal E_{\rm pop}$ quantifies the population change induced by the secular approximation. It also equals
\begin{equation}
\delta\mathcal E_{\rm pop}=\mathcal E_{\rm inc}-\mathcal E^{\rm sec},
\label{eq:population_backaction}
\end{equation}
due to $\mathcal E = \mathcal E_{\rm coh} + \mathcal E_{\rm inc}$.

Figure~\ref{fig:secular_comparison}(a) compares the total ergotropy with and
without full secularization based on the results in Fig.~\ref{fig:detuning}(a). As expected, both predictions vanish at equilibrium.
As the temperature bias $\Delta T$ increases, the partially secular ergotropy grows
continuously from zero, whereas the fully secular state remains passive
until a finite bias activates incoherent ergotropy. Both ergotropies increase, but the partially secular prediction remains larger and their difference widens as the temperature bias increases.

Figure~\ref{fig:secular_comparison}(b) shows that the total difference
$\mathcal E-\mathcal E^{\rm sec}$ closely follows
$\mathcal E_{\rm coh}$ as the temperature bias increases.
The two coincide while the populations of both stationary states retain passive order.
Once population ergotropy appears, the total difference lies slightly below
the coherent contribution. Panel (c) resolves this separation:
$\delta\mathcal E_{\rm pop}$ becomes negative and its magnitude grows with
bias, although it remains small compared with the coherent contribution.
The inset also shows that population ergotropy is activated slightly earlier
in the fully secular state than in the partially secular state. The discrepancy arises because the cross-frequency terms feed coherence back into the
equation of motion through Eq.~\eqref{eq:app_six_equations}, weakening
population inversion and delaying incoherent-ergotropy activation, so the
negative population correction partly offsets the coherent contribution.

\section{Conclusion}
\label{sec:conclusion}

In this work, we have investigated how temperature bias $\Delta T$, qubit detuning $\Delta\omega$, and the exchange coupling $\lambda$ determine the ergotropy of the NESS obtained without the full secular approximation. We proved that the NESS of two resonant qubits with $\Delta\omega=0$ has only coherent ergotropy (Theorem~\ref{thm:resonant_coherence_only}), which grows quadratically with temperature bias near equilibrium at fixed exchange coupling (Eq.~\eqref{eq:temperature_bias_law}) and linearly with exchange coupling in the weak-exchange limit at fixed nonzero temperature bias (Eq.~\eqref{eq:resonant_weak_exchange}). Furthermore, we proved that the NESS of two detuned qubits with $\Delta\omega\neq0$ can have both coherent and incoherent ergotropy (Theorem~\ref{thm:inc-activation}). The incoherent ergotropy is activated only when the hotter reservoir is attached to the higher-frequency bare qubit (Corollary~\ref{cor:activation_direction}), and its exact activation criterion is given in Eq.~\eqref{eq:exact-inc-activation}. Finally, we demonstrated that the ergotropy obtained from the full secular approximation is both qualitatively and quantitatively different from that obtained without the full secular approximation. The ergotropy obtained from the full secular approximation is generally smaller than that obtained without the full secular approximation, and the difference contains both coherent ergotropy and a population correction (Eq.~\eqref{eq:accounting_identity}).

A natural extension of our work is to consider the NESS ergotropy of multi-qubit systems without the full secular approximation. The NESS ergotropy can also be studied beyond the Born--Markov approximation, allowing for strong system--environment interactions~\cite{Leggett1987,DubeStamp1998}. Our work also demonstrates that steady-state coherence plays a key role in ergotropy activation. It would therefore be interesting to explore whether other mechanisms, such as combined longitudinal and transverse qubit--phonon couplings~\cite{Guarnieri2018,Purkayastha2020} or Fano interference induced by polarized incoherent light~\cite{Koyu2021}, can generate steady-state coherence and sustain ergotropy in open quantum systems, and how to optimize ergotropy in these systems. Finally, the physical setting of our work is suitable for experimental tests on current quantum platforms, such as superconducting qubits~\cite{Majer2007,Ronzani2018,PhysRevA.107.023725} and trapped ions~\cite{Barreiro2011}. It would be interesting to experimentally demonstrate the qualitative and quantitative differences between the NESS ergotropy obtained with and without the full secular approximation.

\begin{acknowledgments}
This work was supported by the NSFC (Grant Nos.~12275215 and 12305028) and the Youth Innovation Team of Shaanxi Universities. KZ is supported by the China Postdoctoral Science Foundation under Grant No.~2025M773421, Shaanxi Province Postdoctoral Science Foundation under Grant No.~2025BSHYDZZ017, and Scientific Research Program Funded by Education Department of Shaanxi Provincial Government (Program No.~24JP186).
\end{acknowledgments}
\appendix
\section{Complete analytical stationary solution}
\label{app:br_steady_state}

We construct the NESS analytically following
Refs.~\cite{Wang2019,Zhang2020}. Let $x=\operatorname{Re}c$ and
$y=\operatorname{Im}c$. Substituting the transition operators in
Eq.~\eqref{eq:jump_operators} into Eq.~\eqref{eq:BR_master_equation} gives the
coupled differential equations
\begin{widetext}
\begin{subequations}
\label{eq:app_six_equations}
\begin{align}
&\dot p_1={}2J\bigl(-(\tilde n_1+\tilde n_2)p_1+(1+\tilde n_1)p_2+(1+\tilde n_2)p_3+(\delta n_1+\delta n_2)x\bigr),\\
&\dot p_2={}2J\bigl(\tilde n_1p_1-(1+\tilde n_1+\tilde n_2)p_2+(1+\tilde n_2)p_4+(\delta n_1-\delta n_2)x\bigr),\\
&\dot p_3={}2J\bigl(\tilde n_2p_1-(1+\tilde n_1+\tilde n_2)p_3+(1+\tilde n_1)p_4-(\delta n_1-\delta n_2)x\bigr),\\
&\dot p_4={}2J\bigl(\tilde n_2p_2+\tilde n_1p_3-(2+\tilde n_1+\tilde n_2)p_4-(\delta n_1+\delta n_2)x\bigr),\\
&\dot x={}-2AJx-\Omega y+J\delta n_1(p_1-p_2+p_3-p_4)+J\delta n_2(p_1+p_2-p_3-p_4),\\
&\dot y={}\Omega x-2AJy.
\end{align}
\end{subequations}
Here the parameters $A$, $\tilde n_j$, and $\delta n_j$ are defined in
Eqs.~\eqref{eq:A_definition}, \eqref{eq:ntilde1}--\eqref{eq:ntilde2}, and
\eqref{eq:delta_n_main}, respectively. Setting the six derivatives to zero and
imposing $\sum_{\alpha=1}^4p_\alpha=1$ gives a linear system. Eliminating
$x$ and $y$ first yields Eq.~\eqref{eq:rho23_complex}. Substitution into the
four population equations gives the stationary populations below. This
elimination is the same as that used in Refs.~\cite{Wang2019,Zhang2020} and
introduces no additional approximation.

For brevity, we use $N^{\rm sec}=(1+2\tilde n_1)(1+2\tilde n_2)$ from
Eq.~\eqref{eq:N_sec_definition}. The exact stationary populations are
\begin{subequations}
\label{eq:app_populations}
\begin{align}
&p_1={}\frac{(1+\tilde n_1)(1+\tilde n_2)}{N^{\rm sec}}+\frac{KS(2A+1)}{4A N N^{\rm sec}}\bigl(\delta n_1(1+2\tilde n_1)+\delta n_2(1+2\tilde n_2)\bigr),\\
&p_2={}\frac{\tilde n_1(1+\tilde n_2)}{N^{\rm sec}}+\frac{KS}{4A N N^{\rm sec}}\bigl(\delta n_1(1+2\tilde n_1)(2A-1)-\delta n_2(1+2\tilde n_2)(2A+1)\bigr),\\
&p_3={}\frac{(1+\tilde n_1)\tilde n_2}{N^{\rm sec}}-\frac{KS}{4A N N^{\rm sec}}\bigl(\delta n_1(1+2\tilde n_1)(2A+1)-\delta n_2(1+2\tilde n_2)(2A-1)\bigr),\\
&p_4={}\frac{\tilde n_1\tilde n_2}{N^{\rm sec}}-\frac{KS(2A-1)}{4A N N^{\rm sec}}\bigl(\delta n_1(1+2\tilde n_1)+\delta n_2(1+2\tilde n_2)\bigr),
\end{align}
\end{subequations}
\end{widetext}
where the parameters $N$, $A$, $S$, and $K$ are defined in
Eqs.~\eqref{eq:N_definition}, \eqref{eq:A_definition}, \eqref{eq:Sb}, and
\eqref{eq:K_definition}, respectively.

Each exact population is written as a secular baseline, normalized by
$N^{\rm sec}$, plus a nonsecular correction. Together with the coherence in
Eq.~\eqref{eq:rho23_complex}, the populations in
Eq.~\eqref{eq:app_populations} specify the NESS in Eq.~\eqref{eq:ss_block}.
Direct substitution confirms normalization, while $p_3-p_2$ gives
Eq.~\eqref{eq:Delta_p_analytic}. Setting $\delta n_1=\delta n_2=0$ also
reproduces both the equilibrium and uncoupled limits stated in
Sec.~\ref{subsec:compact_param}. We use the solutions only when $N>0$,
$p_\alpha\ge0$, and $p_2p_3\ge|c|^2$.

Under full secularization, all terms proportional to $\delta n_j$ are removed
from Eq.~\eqref{eq:app_six_equations}. The stationary coherence vanishes and
the analytical state becomes
\begin{equation}
\begin{pmatrix}
p_1^{\rm sec}\\p_2^{\rm sec}\\p_3^{\rm sec}\\p_4^{\rm sec}
\end{pmatrix}
=\frac{1}{N^{\rm sec}}
\begin{pmatrix}
(1+\tilde n_1)(1+\tilde n_2)\\
\tilde n_1(1+\tilde n_2)\\
(1+\tilde n_1)\tilde n_2\\
\tilde n_1\tilde n_2
\end{pmatrix}.
\label{eq:app_secular_populations}
\end{equation}
Equation~\eqref{eq:app_secular_populations} gives
Eq.~\eqref{eq:secular_imbalance}. Since $\tilde n_1,\tilde n_2\ge0$, the first
population in Eq.~\eqref{eq:app_secular_populations} is the largest and the
fourth is the smallest, so only the ordering of the two middle populations can
change. This gives the secular ergotropy in Eq.~\eqref{eq:E_sec_closed} and
completes the analytical connection to Sec.~\ref{sec:secular_comparison}.

\section{Automatic eigenvalue ordering of the physical NESS}
\label{app:automatic_ordering}

We prove analytically that the ordering in Eq.~\eqref{eq:automatic_eigenvalue_ordering}
holds for every physical NESS.

\begin{theorem}
\label{thm:automatic_ordering}
For every physical NESS with $N>0$ and $0<\Omega<2\bar\omega$
whose density matrix $\rho^{\rm ss}$ is positive semidefinite, the four state
eigenvalues satisfy
\begin{equation}
p_1\ge r_+\ge r_-\ge p_4,
\label{eq:app_automatic_ordering}
\end{equation}
where $r_\pm$ are the eigenvalues defined in Eq.~\eqref{eq:rpm}.
\end{theorem}

\begin{proof}
For the analytical stationary state, define
\begin{equation}
D=A^2+\left(\Omega/2J\right)^2,
\label{eq:app_ordering_auxiliary}
\end{equation}
with $A=1+\tilde n_1+\tilde n_2$ given by Eq. \eqref{eq:A_definition}. In addition to the coherence source
\begin{equation}
S=\delta n_1(1+2\tilde n_2)+\delta n_2(1+2\tilde n_1),
\label{eq:app_ordering_source}
\end{equation}
given in Eq. \eqref{eq:Sb},
it is convenient to introduce
\begin{equation}
Q=\delta n_1(1+2\tilde n_1)+\delta n_2(1+2\tilde n_2).
\label{eq:app_ordering_Q}
\end{equation}
The normalization factor and the coefficient $K$ can then be written as
\begin{equation}
N=N^{\rm sec}-\frac{4A^2\delta n_1\delta n_2}{D},
\qquad
K=\frac{A}{D},
\label{eq:app_ordering_NK}
\end{equation}
where $N^{\rm sec}=(1+2\tilde n_1)(1+2\tilde n_2)$ is given by
Eq.~\eqref{eq:N_sec_definition}.

The parameters $S$ and $Q$ obey the exact identity
\begin{equation}
SQ-4A^2\delta n_1\delta n_2
=N^{\rm sec}(\delta n_1-\delta n_2)^2.
\label{eq:app_ordering_SQ_identity}
\end{equation}
It follows that
\begin{equation}
DN+SQ=N^{\rm sec}\left(D+(\delta n_1-\delta n_2)^2\right).
\label{eq:app_ordering_DN_identity}
\end{equation}

We therefore define the positive parameter
\begin{equation}
X=\frac{D+(\delta n_1-\delta n_2)^2}{DN}>0,
\label{eq:app_ordering_X}
\end{equation}
where $D>0$ and $N>0$ for the physical stationary solution considered here.
Using the exact populations in Appendix~\ref{app:br_steady_state}, the sum and
difference of the two outer populations reduce to
\begin{equation}
p_1+p_4=\frac{1+X}{2},
\qquad
p_1-p_4=AX.
\label{eq:app_ordering_outer_sum_difference}
\end{equation}
Equivalently,
\begin{equation}
p_1=\frac{1+(2A+1)X}{4},
\qquad
p_4=\frac{1-(2A-1)X}{4},
\label{eq:app_ordering_outer_populations}
\end{equation}
and normalization gives
\begin{equation}
p_2+p_3=\frac{1-X}{2}.
\label{eq:app_ordering_middle_sum}
\end{equation}

Next, the eigenvalues of the single-excitation subspace are given by
\begin{equation}
r_\pm=\frac{p_2+p_3\pm R}{2},
\label{eq:app_ordering_block_eigenvalues}
\end{equation}
with
\begin{equation}
R=\sqrt{(p_3-p_2)^2+4|c|^2}.
\label{eq:app_ordering_R}
\end{equation}
The population imbalance is
\begin{equation}
p_3-p_2=
\frac{D(\tilde n_2-\tilde n_1)
+A\left((\delta n_2)^2-(\delta n_1)^2\right)}{DN},
\label{eq:app_ordering_imbalance}
\end{equation}
and
\begin{equation}
4|c|^2=\frac{S^2}{DN^2}.
\label{eq:app_ordering_coherence}
\end{equation}
The following algebraic identity is exact:
\begin{multline}
(D+(\delta n_1-\delta n_2)^2)
(A^2(D+(\delta n_1-\delta n_2)^2)-DN)\\
=(D(\tilde n_2-\tilde n_1)
+A((\delta n_2)^2-(\delta n_1)^2))^2+DS^2.
\label{eq:app_ordering_square_identity}
\end{multline}
Combining these relations with the definition of $X$ gives
\begin{equation}
R^2=A^2X^2-X.
\label{eq:app_ordering_R_identity}
\end{equation}

The remaining step follows directly from physicality. Since
$\rho^{\rm ss}$ is positive semidefinite, $p_4\ge0$. From
Eq.~\eqref{eq:app_ordering_outer_populations},
\begin{equation}
(2A-1)X\le1.
\label{eq:app_ordering_physicality_bound}
\end{equation}
Because $\tilde n_1,\tilde n_2\ge0$, one has
$A=1+\tilde n_1+\tilde n_2\ge1$. Hence
\begin{align}
R^2-(A-1)^2X^2
&=X\left((2A-1)X-1\right)\le0,
\label{eq:app_ordering_R_comparison}
\end{align}
and therefore
\begin{equation}
R\le(A-1)X.
\label{eq:app_ordering_R_bound}
\end{equation}
Using Eq.~\eqref{eq:app_ordering_middle_sum}, one finally obtains
\begin{subequations}
\begin{align}
&r_- -p_4=\frac{(A-1)X-R}{2}\ge0,\\
&p_1-r_+=\frac{(A+1)X-R}{2}\ge X>0.
\label{eq:app_ordering_margins}
\end{align}
\end{subequations}
Together with $r_+\ge r_-$ by definition, this proves
\begin{equation}
p_1>r_+\ge r_-\ge p_4
\label{eq:app_ordering_strict_result}
\end{equation}
for every physical NESS described by the analytical solution.
\end{proof}

\section{Temperature-bias and resonant weak-exchange laws}
\label{app:asymptotics}

This appendix derives Eqs.~\eqref{eq:temperature_bias_law} and
\eqref{eq:resonant_weak_exchange} in Sec. \ref{sec:resonance}. Both expansions use a physical, normalized
stationary state with $N>0$. The required eigenvalue ordering follows from
Appendix~\ref{app:automatic_ordering}. We use
$n_j(\omega)=n(\omega,T_j)$ throughout, with the subscript $0$ denoting
the zeroth-order value in the corresponding expansion.

\subsection{Temperature-bias response near equilibrium}
\label{app:near_equilibrium}

We derive Eq.~\eqref{eq:temperature_bias_law} at a fixed, nondegenerate
Hamiltonian. Fix $0<\Omega<2\bar\omega$, $J>0$, and an equilibrium
temperature $T>0$, with $T_1=T$ held fixed and $T_2=T+\Delta T$.
The expansion parameter is $\Delta T=T_2-T_1$, with
$|\Delta T|/T\ll1$.
At zero bias $\Delta T=0$, we have the Bose occupation numbers $n_1(\omega) = n_2(\omega) = n(\omega,T)$, which gives
\begin{subequations}
\begin{align}
&A_0=1+n_1(\omega'_1)+n_1(\omega'_2), \\
&N_0=(1+2n_1(\omega'_1))(1+2n_1(\omega'_2)). 
\end{align}
\end{subequations}
Furthermore, we have the expansions
\begin{subequations}
\begin{align}
&\delta n_j=-\frac{\sin\theta}{2}\,g_j\Delta T+O((\Delta T)^2),\\
&S=-\frac{\sin\theta}{2}\,F_T\Delta T+O((\Delta T)^2),\\
\end{align}
\end{subequations}
where
\begin{subequations}
\begin{align}
&g_j=\partial_T n(\omega'_j,T) = \frac{\omega'_j}{T^2}n_1(\omega'_j)(1+n_1(\omega'_j)),\\
&F_T=g_1(1+2n_1(\omega'_2))+g_2(1+2n_1(\omega'_1)). 
\end{align}
\end{subequations}
Then the NESS coherence becomes
\begin{align}
|c|^2=\frac{\sin^2\theta\,F_T^2}
{4N_0^2(4A_0^2+(\Omega/J)^2)}
(\Delta T)^2+O((\Delta T)^3).
\label{eq:app_bias_source_expansion}
\end{align}
Here $\Delta p=(n_1(\omega'_2)-n_1(\omega'_1))/N_0+O(\Delta T)<0$ and
$N=N_0+O(\Delta T)$. Since the finite-temperature Gibbs state has
strictly ordered positive eigenvalues, a sufficiently small bias remains
within the physical stationary-state neighborhood. Expanding the compact expression
gives $\mathcal E_{\rm coh}=\Omega|c|^2/|\Delta p|+O(|c|^4)$ and hence
\begin{equation}
\mathcal E_{\rm coh}=C_T(\Delta T)^2+O((\Delta T)^3),
\label{eq:app_bias_quadratic_law}
\end{equation}
where the coefficient is given by
\begin{equation}
C_T=\frac{\Omega\sin^2\theta\,F_T^2}
{4N_0(n_1(\omega'_1)-n_1(\omega'_2))(4A_0^2+(\Omega/J)^2)}.
\end{equation}
This proves Eq.~\eqref{eq:temperature_bias_law}.  

\subsection{Weak exchange-coupling law}
\label{app:resonant_weak_exchange}

To derive Eq.~\eqref{eq:resonant_weak_exchange}, we set $\Delta\omega=0$
before taking $\lambda\to0^+$, keeping
$\bar\omega,J>0$ and unequal bath temperatures fixed. The two
transition frequencies coalesce at $\bar\omega$ in this
limit, while $\theta=\pi/2$ on the nondegenerate side. Consequently, the parameters $N$ in Eq. \eqref{eq:N_definition}, $A$ in Eq.~\eqref{eq:A_definition}, $S$ in Eq. \eqref{eq:Sb}, and $K$ in Eq. \eqref{eq:K_definition} have the expansions $N=N_0+O(\lambda^2)$, $A=A_0+O(\lambda^2)$, $S=S_0+O(\lambda^2)$, and $K = K_0 +O(\lambda^2)$ with
\begin{subequations}
\begin{align}
&N_0 =A_0^2-(n_1(\bar\omega)-n_2(\bar\omega))^2, \\
&A_0 =1+n_1(\bar\omega)+n_2(\bar\omega), \\
&S_0 = (n_1(\bar\omega)-n_2(\bar\omega))A_0, \\
&K_0 = \frac{1}{A_0}.
\end{align}
\end{subequations}
The linear corrections cancel because the two
transition frequencies shift symmetrically about $\bar\omega$.
Therefore, the NESS coherence magnitude becomes
\begin{equation}
|c| = \frac{|n_1(\bar\omega)-n_2(\bar\omega)|}{2N_0} + O(\lambda^2).
\end{equation}
Substituting it and $\Delta p=O(\lambda)$ into
Eq.~\eqref{eq:Ecoh_raw} gives
\begin{equation}
\mathcal E_{\rm coh}
=\frac{|n_1(\bar\omega)-n_2(\bar\omega)|}{2N_0}\lambda+O(\lambda^2).
\end{equation}
Using $N_0=(1+2n_1(\bar\omega))(1+2n_2(\bar\omega))$ yields
\begin{equation}
\mathcal E_{\rm coh}
=\frac{|n_1(\bar\omega)-n_2(\bar\omega)|}{2(1+2n_1(\bar\omega))(1+2n_2(\bar\omega))}\lambda
+O(\lambda^2).
\label{eq:app_resonant_linear_law}
\end{equation}
This is Eq.~\eqref{eq:resonant_weak_exchange}.
Thus $|c|\to|n_1(\bar\omega)-n_2(\bar\omega)|/(2N_0)$ is finite, whereas the energy
splitting available to extract work tends to zero. At equilibrium the
coefficient vanishes and the Gibbs state has zero ergotropy for every
exchange coupling. 

\section{Weak-exchange coefficients at finite detuning}
\label{app:finite_detuning_analysis}

In this appendix, we derive the coefficients $C_{\Delta\tilde p}$,
$C_{\Delta \tilde p\neq 0}$, and $C_{\Delta \tilde p= 0}$ in
Eqs.~\eqref{eq:weak_exchange_population_order},
\eqref{eq:detuned_generic_condition}, and
\eqref{eq:detuned_tuned_coherent_law}, respectively. Hold
$\Delta\omega>0$, $T_2>T_1>0$, and $J>0$ fixed, and take
$\lambda/\Delta\omega\to0^+$. We use $n_j(\omega)=n(\omega,T_j)$. In the uncoupled limit, the parameters have the expansions $N=N_0+O(\lambda^2)$, $A=A_0+O(\lambda^2)$, and $K = K_0 +O(\lambda^2)$ with
\begin{subequations}
\label{eq:app_detuned_coeff_definitions}
\begin{align}
&N_0=(1+2n_1(\omega_1))(1+2n_2(\omega_2)),\\
&A_0=1+n_1(\omega_1)+n_2(\omega_2),\\
&K_0=\frac{A_0}{A_0^2+(\Delta\omega/2J)^2}.
\end{align}
\end{subequations}
At zero exchange, we have
$\Delta\tilde p(\lambda=0)=n_2(\omega_2)-n_1(\omega_1)$,
$\Delta p(\lambda=0)=\Delta\tilde p(\lambda=0)/N_0$,
and $c(\lambda=0)=0$. The frequency derivatives satisfy
\begin{equation}
-\partial_\omega n_j(\omega)=\frac{n_j(\omega)(1+n_j(\omega))}{T_j}>0.
\end{equation}
The argument in Corollary~\ref{cor:activation_direction} shows that
$n_2(\omega)-n_1(\omega)$ is positive and decreases strictly with frequency.

Expanding the splitting and mixing weights gives
\begin{subequations}
\label{eq:app_population_frequency_mixing_expansion}
\begin{align}
&\Omega=\Delta\omega+\frac{\lambda^2}{2\Delta\omega}+O(\lambda^4),\\
&\omega'_{1,2}=\omega_{1,2}\mp\frac{\lambda^2}{4\Delta\omega}+O(\lambda^4),\\
&\cos^2(\theta/2)=1-\frac{\lambda^2}{4(\Delta\omega)^2}+O(\lambda^4),\\
&\sin^2(\theta/2)=\frac{\lambda^2}{4(\Delta\omega)^2}+O(\lambda^4),
\end{align}
\end{subequations}
Together with $\sin\theta=\lambda/\Delta\omega+O(\lambda^3)$,
substitution into Eqs.~\eqref{eq:ntilde1}--\eqref{eq:delta_n_main} yields
\begin{subequations}
\label{eq:app_population_occupation_expansion}
\begin{align}
&\begin{aligned}[t]
\tilde n_1={}&n_1(\omega_1)
-\frac{\partial_\omega n_1(\omega_1)}{4\Delta\omega}\lambda^2\\
&+\frac{n_2(\omega_1)-n_1(\omega_1)}{4(\Delta\omega)^2}\lambda^2
+O(\lambda^4),
\end{aligned}\\
&\begin{aligned}[t]
\tilde n_2={}&n_2(\omega_2)
+\frac{\partial_\omega n_2(\omega_2)}{4\Delta\omega}\lambda^2\\
&-\frac{n_2(\omega_2)-n_1(\omega_2)}{4(\Delta\omega)^2}\lambda^2
+O(\lambda^4),
\end{aligned}\\
&\delta n_1=-\frac{n_2(\omega_1)-n_1(\omega_1)}{2\Delta\omega}\lambda+O(\lambda^3),\\
&\delta n_2=-\frac{n_2(\omega_2)-n_1(\omega_2)}{2\Delta\omega}\lambda+O(\lambda^3).
\end{align}
\end{subequations}
Inserting the above expansions into the numerator of
Eq.~\eqref{eq:Delta_p_analytic},
$\Delta\tilde p=\tilde n_2-\tilde n_1+K((\delta n_2)^2-(\delta n_1)^2)$,
gives
\begin{equation}
\Delta\tilde p=\Delta\tilde p(\lambda=0)-C_{\Delta\tilde p}\lambda^2+O(\lambda^4),
\label{eq:app_weak_population_order}
\end{equation}
with
\begin{multline}
C_{\Delta\tilde p}=-\frac{\partial_\omega n_1(\omega_1)+\partial_\omega n_2(\omega_2)}{4\Delta\omega}\\
+\frac{n_2(\omega_1)-n_1(\omega_1)+n_2(\omega_2)-n_1(\omega_2)}{4(\Delta\omega)^2}\\
+\frac{K_0}{4(\Delta\omega)^2}\bigl((n_2(\omega_1)-n_1(\omega_1))^2\\
-(n_2(\omega_2)-n_1(\omega_2))^2\bigr)>0.
\label{eq:app_Cpop}
\end{multline}
Note that the last term arises from the nonsecular population correction. This establishes
Eq.~\eqref{eq:weak_exchange_population_order}.

To obtain the coherent-ergotropy coefficients, define
\begin{multline}
\mathcal F=(n_2(\omega_1)-n_1(\omega_1))(1+2n_2(\omega_2))\\
+(n_2(\omega_2)-n_1(\omega_2))(1+2n_1(\omega_1))>0.
\label{eq:app_detuned_coherence_source}
\end{multline}
Equations~\eqref{eq:Sb} and~\eqref{eq:rho23_complex} then give
\begin{subequations}
\label{eq:app_detuned_source_expansion}
\begin{align}
&S=-\frac{\mathcal F}{2\Delta\omega}\lambda+O(\lambda^3),\\
&|c|=\frac{\mathcal F}
{2\Delta\omega N_0\sqrt{4A_0^2+(\Delta\omega/J)^2}}\lambda+O(\lambda^3).
\end{align}
\end{subequations}
For $\Delta\tilde p(\lambda=0)\ne0$,
$|\Delta p|=|\Delta\tilde p(\lambda=0)|/N_0+O(\lambda^2)$,
so expanding the
square root in Eq.~\eqref{eq:Ecoh_raw} gives
\begin{equation}
\mathcal E_{\rm coh}
=\frac{\Omega|c|^2}{|\Delta p|}+O(\lambda^4).
\label{eq:app_detuned_quadratic_law}
\end{equation}
Substituting Eq.~\eqref{eq:app_detuned_source_expansion} gives the coefficient
in Eq.~\eqref{eq:detuned_generic_condition}:
\begin{equation}
C_{\Delta \tilde p\neq 0}
=\frac{\mathcal F^2}
{4\Delta\omega N_0|\Delta\tilde p(\lambda=0)|
(4A_0^2+(\Delta\omega/J)^2)}.
\label{eq:app_Ccoh_generic}
\end{equation}

At the tuned balance $\Delta\tilde p(\lambda=0)=0$, Eq.~\eqref{eq:app_weak_population_order}
instead gives
$\Delta p=-C_{\Delta\tilde p}\lambda^2/N_0+O(\lambda^4)$.
Equation~\eqref{eq:app_detuned_source_expansion} then gives
$\sqrt{(\Delta p)^2+4|c|^2}=2|c|+O(\lambda^3)$.
Equation~\eqref{eq:Ecoh_raw} therefore becomes
\begin{equation}
\mathcal E_{\rm coh}
=\Delta\omega|c|
-\frac{\Delta\omega C_{\Delta\tilde p}}{2N_0}\lambda^2+O(\lambda^3),
\label{eq:app_detuned_tuned_linear_law}
\end{equation}
which yields the coefficient in Eq.~\eqref{eq:detuned_tuned_coherent_law}:
\begin{equation}
C_{\Delta \tilde p= 0}=\frac{\mathcal F}{2N_0\sqrt{4A_0^2+(\Delta\omega/J)^2}}>0,
\label{eq:app_Ccoh_tuned}
\end{equation}
with all parameters evaluated at $\Delta\tilde p(\lambda=0)=0$. These expansions assume fixed
nonzero detuning $\Delta\omega$.


%

\end{document}